\documentclass[11pt]{article}
\usepackage[letterpaper,margin=1in]{geometry}
\usepackage{microtype}
\usepackage[font=small,labelfont=bf]{caption}
\usepackage{amsmath, amsfonts, amssymb}
\usepackage{mathtools}
\usepackage{bm}
\usepackage{graphicx}
\usepackage{booktabs}
\usepackage{algorithm}
\usepackage{algpseudocode}
\usepackage[round]{natbib}
\usepackage{float}
\usepackage{enumitem}
\usepackage{amsthm}
\usepackage{hyperref}
\hypersetup{colorlinks=true,linkcolor=blue,citecolor=blue,urlcolor=blue}

\newtheorem{assumption}{Assumption}
\newtheorem{proposition}{Proposition}
\newtheorem{theorem}{Theorem}
\newtheorem{condition}{Condition}
\newtheorem{corollary}{Corollary}

\newtheorem{remark}{Remark}
\newcommand{\EE}{\mathbb{E}}
\newcommand{\PP}{\mathbb{P}}

\newcommand{\norm}[1]{\left\lVert#1\right\rVert}
\newcommand{\bX}{\boldsymbol{X}}
\newcommand{\RR}{\mathbb{R}}

\title{Debiased Machine Learning for Count Data: a Partially Linear Poisson Model Based on Neural Networks}
\author{%
Ningkun Zhou$^{1}$ \quad Bryan E. Shepherd$^{1}$ \quad Qingyan Xiang$^{2}$\\[0.5em]
{\small $^{1}$Vanderbilt University}\\
{\small $^{2}$Vanderbilt University Medical Center}
}
\date{}

\begin{document}
\raggedbottom
\maketitle

\begin{abstract}
We develop a novel debiased machine learning (DML) estimator to analyze count data, which are common in clinical or biomedical studies. 
Specifically, we estimate the effect of a binary treatment on count outcomes using a partially linear Poisson model. 
A key advantage of this model is that it represents the complex covariate structure through a nuisance function, which is estimated flexibly using neural networks.
We use a Neyman-orthogonal score function to construct the estimator, reducing its sensitivity to errors in estimating the nuisance function.  
Cross-fitting is used to mitigate overfitting bias in neural networks.
Under mild regularity conditions, this DML estimator is asymptotically normal with root-$n$ convergence. 
We derive a closed-form variance estimator and construct a Wald confidence interval for the treatment effect.
Extensive simulations demonstrate that the proposed estimator reduces bias and root mean square error compared with the generalized linear model and the augmented inverse probability weighting (AIPW) estimator in most settings, while achieving a coverage probability near the nominal level.
The proposed procedure is implemented in the R package \texttt{PoissonDML}. 
We apply our approach to a synthetic HIV cohort data to investigate the effect of treatment on AIDS-defining event counts of 7,141 analyzed patients.
Our approach provides practical and valid estimation and inference for treatment effects under the partially linear Poisson model while combining modern machine learning methods to estimate complex covariate structures.
\end{abstract}

\section{Introduction} \label{sec:intro}

\subsection{Background and motivation}
Count outcomes commonly appear in epidemiological or clinical studies because patients may experience recurrent events, such as infections, hospitalizations, or adverse events. 
When the outcomes are counts, researchers are often interested in comparing the expected event counts between treatments, which can inform assessments of disease progression and treatment decisions. 
For example, in a large clinical HIV cohort \citep{mcgowan07ccasanet}, it is of interest to estimate the effect of dolutegravir (DTG)-based versus non-DTG antiretroviral therapy (ART) on the subsequent number of AIDS-defining events (ADEs), while controlling for demographic and clinical baseline variables.
As the presence of multiple confounders may introduce bias into effect estimates, it is essential to develop proper methods to analyze such data.

A standard approach to analyzing count outcomes is Poisson regression within the generalized linear model (GLM) framework. 
However, Poisson regression relies on correct specification of the parametric model of covariates, which can be challenging in observational studies with complex covariate structures. 
This challenge motivates semiparametric methods that aim to flexibly estimate the relationship between the covariates and the outcome. 
For example, the partially linear regression model proposed by \cite{robinson88plr} allows the covariate effects to be estimated flexibly while preserving an interpretable treatment effect parameter. 

\subsection{Literature review}

Modern machine learning methods have increasingly been adopted to model data with high-dimensional and complex covariate structures.
Among those approaches, neural networks have received particular attention because of their ability to approximate complex nonlinear functions and complex interactions without requiring explicit model specification \citep{hornik89mlp, hornik91mlp}. 
For example, several recent studies have adapted neural network architectures for count data by integrating Poisson or related likelihood functions \citep{montesinos20dlPoisson,montesinos21dlPoisson,jose24NNcounts}.
However, these methods are primarily designed for prediction, not estimation or inference on treatment effect parameters. 
It is challenging to directly incorporate machine-learned predictions into estimating equations, as these predictions are subject to bias. 
To address this, semiparametric frameworks have been developed to incorporate flexible machine-learned predictions in nuisance models for the augmented inverse probability weighting estimator (AIPW) \citep{robins94AIPW} and targeted maximum likelihood estimation (TMLE) \citep{vanderLaanRubin06tmle}.

More recently, \cite{chernozhukov18double} developed the double/debiased machine learning (DML) framework, which combines machine-learned nuisance functions with Neyman-orthogonal scores and cross-fitting.
These components reduce the sensitivity of estimating equations to errors in nuisance estimation and overfitting, enabling root-$n$ consistency and asymptotic normality under suitable conditions.
Subsequent developments address multiway-clustered data \citep{chiang22multiway}, heterogeneous effects and other causal functions \citep{semenova21cate}, mediation and dynamic treatments \citep{farbmacher22mediation,bodory22dynamic}, multiple and continuous treatments \citep{xiang25multiple,colangelo26continuous}, and nonlinear semiparametric models \citep{liu21logistic,nekipelov22orthogonal}.

Several studies have developed semiparametric estimating equations for Poisson models.
\cite{robinsRotnitzky01cmt} described a general class of estimating functions for regression models with a log link.
The Stata package \texttt{xpopoisson} implements Poisson regression using lasso-based nuisance estimation \citep{stata25xpopoisson}.
\cite{covert23relinquishing} derived a cross-fitted pseudo-Poisson estimator for binary treatments using the conditional outcome mean and the propensity score.
\cite{kumar22pricing} developed an orthogonal procedure for continuous price-sensitivity estimation in the presence of treatment heterogeneity. 
However, the existing literature has not provided a complete DML framework for this setting, including a practical algorithm for estimation, asymptotic theory for the DML estimator, and a variance estimator for constructing valid Wald confidence intervals.

\subsection{Our contribution}
To our knowledge, this work is the first to develop a practical estimation and inference procedure for the DML estimator under the partially linear Poisson model and to provide a comprehensive theoretical justification for its asymptotic properties.
Specifically, our main contributions are that we 
(1) develop a cross-fitted DML estimator that directly estimates $g(X)$ using a neural network; 
(2) establish Neyman orthogonality, double robustness, and root-$n$ asymptotic normality of the proposed estimator; 
(3) develop the R package \texttt{PoissonDML} to implement the proposed estimation and inference procedure, 
and (4) apply the proposed DML method to a synthetic HIV cohort to evaluate treatment effects on AIDS-defining event counts.
The R package \texttt{PoissonDML} can be accessed through GitHub  (\url{https://github.com/nkzhou/PoissonDML-R}).

\section{Model setup and estimating equations}
\label{sec:setup}
In this section, we first introduce the partially linear Poisson regression model. 
Then we present our proposed estimating equation and show that it is Neyman orthogonal.

\subsection{Partially linear Poisson model}
For $i=1,2,\ldots,n$, let $W_i=\left(Y_i,A_i,\bX_i\right)$ denote an i.i.d. observation, where $Y_i$ is a count outcome, $A_i\in\{0,1\}$ is a binary treatment, and $\bX_i\in\RR^p$ is a vector of observed covariates.
We assume that the count outcome and binary treatment are generated according to
\begin{align} \label{eq:plpois}
Y_i\mid \left(A_i,\bX_i\right)
    &\sim \mathrm{Poisson} \left(\lambda_0 \left(A_i,\bX_i \right) \right), \nonumber\\
\log\lambda_0 \left(A_i,\bX_i\right)
    &=\theta_0A_i+g_0\left(\bX_i\right),
\end{align}
and 
$$
A_i\mid\bX_i \sim \mathrm{Bernoulli}\left(\pi_0(\bX_i)\right). 
$$ 
Here, $\lambda_0(A_i, \bX_i)=\EE[Y_i\mid A_i, \bX_i]$ is the underlying conditional mean count and $g_0$ is an unknown nuisance function of $\bX_i$. 
The target parameter is the scalar coefficient $\theta_0$.
Since  $A_i\in\{0,1\}$ is a binary treatment, $\pi_0(\bX_i) =\PP\left(A_i=1\mid\bX_i\right)$ is the true propensity score and represents the conditional treatment assignment mechanism given $\bX_i$.

In clinical studies with binary treatments, 
\textbf{the parameter of interest} is often 
$$ e^{\theta_0} = \frac{\EE\left[Y_i\mid A_i=1,\bX_i\right]} {\EE\left[Y_i\mid A_i=0,\bX_i\right]}, 
$$
the conditional count mean ratio comparing two treatment regimens while holding $\bX_i$ fixed.
This parameter can have a causal interpretation under the potential outcome framework in an observational study, which requires appropriate identification assumptions, including consistency, positivity, and exchangeability \citep{hernan20whatif}.

\subsection{The Neyman orthogonal score function}
\label{sec:score}
The nuisance function $g_0$ is estimated by flexible machine learners such as neural networks. 
A naive approach is to jointly estimate the treatment effect and nuisance function using neural networks and then report the fitted treatment coefficient.
However, this approach is generally not appropriate, because estimation errors in the nuisance function $\hat{g}$ can affect the estimation of $\theta_0$. 
As a result, the naive estimator of $\theta_0$ may fail to be root-$n$ consistent and asymptotically normal.

To address this gap, we propose the following score function for $\theta_0$:
\begin{align} \label{eq:score}
    \psi\left(W_i;\theta,\eta\right)=\left(A_i-\mu\left(\bX_i\right)\right)\left(Y_i-e^{\theta A_i+g\left(\bX_i\right)}\right),
\end{align}
where $\theta$ is a candidate value of $\theta_0$, and $\eta=\left(g,\mu\right)$ denotes the nuisance functions to be estimated in the DML framework.
We derive that 
$$
\mu\left(\bX_i\right)=\frac{\EE\left[A_i\exp\{\theta A_i\}\mid\bX_i\right]}{\EE\left[\exp\{\theta A_i\}\mid\bX_i\right]}.
$$
For a binary treatment, this can be written as 
\begin{align} \label{eq:mu}
    \mu\left(\bX_i\right)=\frac{\pi\left(\bX_i\right)e^\theta}{1-\pi\left(\bX_i\right)+\pi\left(\bX_i\right)e^\theta},
\end{align}
where $\pi\left(\bX_i\right)=\PP\left(A_i=1\mid\bX_i\right)$ is the propensity score for observation $i$. 
The derivation of this score function is provided in Appendix~\ref{sec:scoreConstruction}.

Let $\eta_0=(g_0,\mu_0)$ be the true nuisance parameters.
Then under equation~(\ref{eq:plpois}), the proposed score function satisfies the population moment condition
\begin{align} \label{eq:moment}
    \EE\left[\psi\left(W_i;\theta_0,\eta_0\right)\right]=0.
\end{align}
In addition, the population moment satisfies 
\begin{align*}
    &\partial_g\EE[\psi(W_i;\theta_0,\eta)]=0\big|_{\eta=\eta_0} \qquad  \text{and} \qquad  \partial_\mu\EE[\psi(W_i;\theta_0,\eta)]=0\big|_{\eta=\eta_0}.
\end{align*}
They imply that the moment condition in equation~(\ref{eq:moment}) is insensitive to local perturbations of $\eta$ in a neighborhood of $\eta_0$. 
Consequently, even when the flexible learner $\hat{\eta}=\left(\hat{g},\hat{\mu}\right)$ converges at a slower rate, the estimator of $\theta_0$ can still achieve the desired $n^{-1/2}$ rate under suitable regularity conditions.
With these properties, $\psi$ is a Neyman-orthogonal score function for the partially linear Poisson model in equation~(\ref{eq:plpois}), which is formally defined in \cite{chernozhukov18double} and \cite{chiang22multiway}.
The Neyman orthogonality of this score is formally proved in Proposition~\ref{prop:orthogonal}.
In a broader context, our score function can be viewed as a model-specific instance of the general class of semiparametric estimating functions discussed by \cite{robinsRotnitzky01cmt}.

\section{Estimation procedure} \label{sec:procedure}

In this section, we present the estimation of nuisance functions, construction of the estimating equation, and inference for the DML estimator.
The complete algorithm for parameter estimation and inference is provided in Appendix~\ref{sec:dmlalgorithm}.

\subsection{Nuisance parameter estimation}
\label{subsec:nuest}
The proposed score function relies on two nuisance functions: the outcome nuisance $g_0(\bX_i)$ and the treatment nuisance $\mu_0(\bX_i)$.
To estimate $\mu_0(\bX_i)$, we first need to estimate $\pi_0(\bX_i)$ using a propensity score model. 

\begin{figure}[t]
    \centering
    \includegraphics[width=\textwidth]{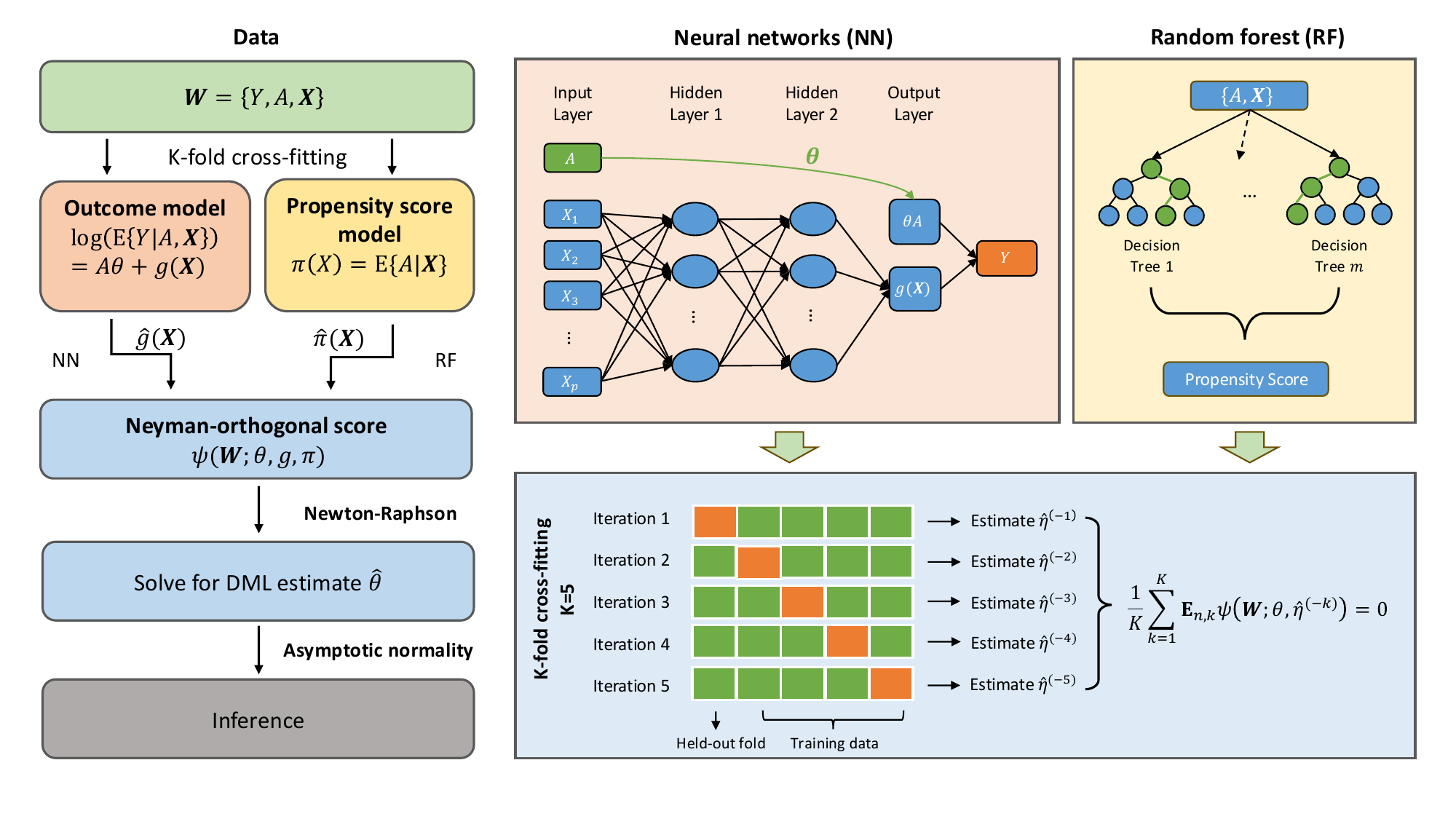}
    \caption{\textbf{Overview of the DML estimation and inference procedure.} The center panel shows the partially linear Poisson neural network for the outcome nuisance, and the right panel shows random-forest estimation of the propensity score. The left and lower panels show how the resulting nuisance estimates enter the Neyman-orthogonal score under cross-fitting and are used to estimate $\theta_0$.}
    \label{fig:dml-flowchart}
\end{figure}

Figure~\ref{fig:dml-flowchart} summarizes this procedure.
The upper middle panel illustrates the neural network (NN) model used to approximate the outcome nuisance function $g_0(\bX_i)$.
The outcome model has two branches. 
In the first branch, the covariates $\bX_i$ pass through a feedforward neural network with multiple hidden layers, allowing $g(\bX_i)$ to model complex nonlinearities and interactions among the covariates.
Here, $g(\bX_i)$ is the approximation to $g_0$.
In the second branch, the binary treatment $A_i$ bypasses the hidden layers and enters a separate linear node, contributing $\theta A_i$ to the output. 
These branches are combined at the output layer, giving the conditional mean count $\lambda_{\theta,g}(A_i,\bX_i)=\exp\{\theta A_i+g(\bX_i)\}$. 
This architecture allows a flexible structure for $g_0(\bX_i)$ while maintaining the interpretable treatment-effect parameter.
The Poisson negative log-likelihood is used as the loss function to train the neural networks and takes the form 
$$
-\log\mathcal L\left(W_i;\theta,g\right)=\exp\{\theta A_i+g\left(\bX_i\right)\}-Y_i\left(\theta A_i+g\left(\bX_i\right)\right),
$$
where terms independent of both $\theta$ and $g(\bX_i)$ are omitted. 

For the treatment nuisance function, let $\pi(\bX_i)$ denote an approximation to the true propensity score $\pi_0(\bX_i)=\PP(A_i=1\mid\bX_i)$, which can be estimated using different methods such as logistic regression, bagging or boosting, and recursive partitioning or tree-based methods \citep{austin11ps}. 

For estimation of both nuisance functions, we use $K$-fold cross-fitting to reduce overfitting bias \citep{chernozhukov18double} in machine learners, as illustrated in the lower panel of Figure~\ref{fig:dml-flowchart}.
In each fold, the outcome and treatment models are trained using observations outside the held-out fold and then used to predict $g_0(\bX_i)$ and $\pi_0(\bX_i)$ for observations in that fold.
After pooling the held-out predictions across folds, we obtain the cross-fitted nuisance estimates $\hat g_i$ and $\hat\pi_i$ for observation $i$, which are then used to construct the Neyman-orthogonal estimating equation.

\subsection{Estimating equation and Newton--Raphson root search}
Substituting $\hat g_i$ and $\hat\pi_i$ into equations~(\ref{eq:plpois}) and~(\ref{eq:mu}), respectively, we obtain 
\begin{align*}
    \hat{\lambda}_i\left(\theta\right)=e^{\theta A_i+\hat g_i} \qquad \text{and} \qquad \hat\mu_i(\theta)=\frac{\hat\pi_i e^{\theta}}{\left(1-\hat\pi_i+\hat\pi_i e^{\theta} \right)}.
\end{align*}
Then, substituting these into equation~\ref{eq:score} gives the fitted score function 
$$
\hat\psi_i(\theta) = (A_i-\hat\mu_i(\theta))(Y_i-\hat\lambda_i(\theta)). 
$$
Following the procedure of constructing a DML2 estimator in Definition~3.2 of \cite{chernozhukov18double}, the DML estimator $\hat\theta$ is the solution to the estimating equation $n^{-1}\sum_{i=1}^n\hat\psi_i\left(\theta\right)=0$.

We use the Newton--Raphson algorithm to solve the estimating equation for $\theta$.
The iterations are updated according to
\begin{align*}
    \theta^{(m+1)} = \theta^{(m)} -
    \frac{n^{-1}\sum_{i=1}^n \hat\psi_i\left(\theta^{(m)}\right)}
    {n^{-1}\sum_{i=1}^n \hat\psi_i^\prime\left(\theta^{(m)}\right)}.
\end{align*}
Here, 
\begin{align*}
    \hat\psi_i^\prime\left(\theta\right)
      =-\hat\mu_i^\prime\left(\theta\right)
        \left(Y_i-\hat\lambda_i\left(\theta\right)\right)
        -\left(A_i-\hat\mu_i\left(\theta\right)\right)
        A_i\hat\lambda_i\left(\theta\right)
\end{align*}
is the derivative of the score function, and 
$$
\hat\mu_i^\prime\left(\theta\right)=\hat\mu_i\left(\theta\right)\left(1-\hat\mu_i\left(\theta\right)\right)
$$ 
for a binary treatment.

\subsection{Variance estimator and Wald confidence intervals}
\label{subsec:wald}
We follow Theorem 3.3 and Corollary 3.1 of \cite{chernozhukov18double} and Theorem 4.2 of \cite{liu21logistic} to construct a variance estimator and a Wald confidence interval for our DML estimator.
The asymptotic validity of this inference procedure is established formally in Theorem~\ref{thm:normal}. 
Here, we describe how the asymptotic variance is estimated in practice.

We first define the Jacobian with respect to $\theta$, evaluated at $\theta_0$, as 
$$
J_0=\partial_\theta\EE\left[\psi\left(W_i;\theta,\eta_0\right)\right]\big|_{\theta=\theta_0}. 
$$
For our orthogonal score function, this Jacobian becomes
\begin{align}
    J_0 
    =\EE&\left[-\mu_0\left(\bX_i\right)\left(1-\mu_0\left(\bX_i\right)\right)
    \left(Y_i-\lambda_0(A_i,\bX_i)\right) 
    -\left(A_i-\mu_0(\bX_i)\right)
    A_i\lambda_0(A_i,\bX_i)\right]. \label{eq:J0}
\end{align}
Combining the Jacobian and the second moment of the score, we obtain the corresponding asymptotic variance 
\begin{align} \label{eq:asyvar}
    \sigma_0^2 = J_0^{-2} \EE\left[\psi^2 \left(W_i; \theta_0, \eta_0 \right) \right]. 
\end{align}

Substituting the machine-learned estimates $\hat g_i$ and $\hat\pi_i$ and evaluating at the DML estimate $\hat\theta$ give
$\hat\lambda_i(\hat\theta)=\exp\{\hat\theta A_i+\hat g_i\}$ and
$$
\hat\mu_i(\hat\theta)=\frac{\hat\pi_i e^{\hat\theta}}{1-\hat\pi_i+\hat\pi_i e^{\hat\theta}}.
$$
The Jacobian in equation~(\ref{eq:J0}) is then estimated by the sample average of individual score derivatives:
\begin{align*}
    \hat J =n^{-1}\sum_{i=1}^n
            -\hat\mu_i(\hat\theta)
            \left(1-\hat\mu_i(\hat\theta)\right) \left(Y_i-\hat\lambda_i(\hat\theta)\right)
            -
            \left(A_i-\hat\mu_i(\hat\theta)\right)
            A_i\hat\lambda_i(\hat\theta).
\end{align*}
Next, we estimate the population second moment $\EE[\psi^2(W_i;\theta_0,\eta_0)]$ by $n^{-1}\sum_{i=1}^n\psi^2(W_i;\hat\theta,\hat\eta)$.
Thus, the asymptotic variance in equation~(\ref{eq:asyvar}) is estimated by
\begin{align*}
    \hat\sigma^2
    &=\hat{J}^{-2}\,n^{-1}\sum_{i=1}^n
    \left[\left(A_i-\hat\mu_i(\hat\theta)\right)
        \left(Y_i-\hat\lambda_i(\hat\theta)\right)\right]^2.
\end{align*}

It follows that the estimated variance and standard error of $\hat\theta$ are 
\begin{align*}
    \widehat{\operatorname{Var}}(\hat\theta)=\frac{\hat\sigma^2}n \qquad \text{and} \qquad \widehat{\operatorname{se}} (\hat\theta)=\sqrt{\frac{\hat\sigma^2}n}.
\end{align*}
Therefore, a Wald confidence interval with nominal coverage $1-\alpha$ is $\hat\theta\pm z_{1-\alpha/2}\,\widehat{\operatorname{se}}(\hat\theta)$,
where $z_{1-\alpha/2}$ is the standard normal quantile.
This procedure, including nuisance parameter fitting, parameter estimation, and statistical inference, is implemented in the R package \texttt{PoissonDML}. 

\section{Theory} \label{sec:theory}
In this section, we establish that the proposed estimating equation is Neyman orthogonal and derive the asymptotic properties of the resulting DML estimator needed for valid statistical inference.
First, we introduce some notation. 
Let $o(a_n)$, $O(a_n)$, and $\Omega(a_n)$ represent sequences growing at a smaller, smaller/equal, and greater/equal rate of $a_n$, respectively. 
Also, let $o_p$, $O_p$, and $\Omega_p$ be the corresponding rates with probability approaching 1 as $n\rightarrow\infty$. 
In addition, let $\mathcal{X} \subseteq \mathbb{R}^p$ denote the domain of $\bX_i$.

We state the following regularity conditions for $\theta$ and the score function used for the DML estimator, adapted from Assumption REG in \cite{liu21logistic}.
\begin{assumption}[Regularity conditions]
\label{ass:reg}
The following regularity conditions hold:
\begin{enumerate}[label=(\alph*)]
    \item The parameter $\theta$ belongs to a compact set $\Theta \subseteq \mathbb{R}$ and there exists a sequence of positive constants $\delta_n = \Omega(n^{-1/2}\log n)$ such that $(\theta_0 - \delta_n, \theta_0 + \delta_n) \subseteq \Theta$.
    
    \item The outcome $Y_i$ satisfies finite moment conditions: $\EE [|Y_i|^q] < \infty$ for some $q > 2$.
    
    \item The nuisance functions $g_0$ and $\pi_0$ are uniformly bounded.
    Specifically, there exists a constant $C$ such that $\norm{g_0}_\infty=\sup_{\bX\in\mathcal{X}}|g_0(\bX)|\leq C$ and the propensity score $\pi_0(\bX)$ is bounded away from 0 and 1.
    
    \item There exist constants $0 < c_0 \leq c_1 < \infty$ such that $c_0 \leq |J_0| \leq c_1$ and $c_0\leq\EE[\psi^2(W_i;\theta_0, \eta_0)]\leq c_1$, where $J_0$ is defined in equation~(\ref{eq:J0}).
\end{enumerate}
\end{assumption}

We further define $\|f(U)\|_{Q,q} = \left(\int |f(u)|^q \,dQ(u)\right)^{1/q}$ for any real number $q \geq1$, measurable function $f$, random variable $U$, and probability measure $Q$.
With this notation, we adapt Assumption ML1 from \cite{liu21logistic} and introduce the following assumption:
\begin{assumption}[Quality of ML Nuisance Estimators]
\label{ass:ml}
Under $K$-fold cross-fitting, the nuisance models are trained on a training sample and evaluated on the corresponding held-out sample.
Let $\hat g(\bX)$ and $\hat\pi(\bX)$ denote the nuisance estimators fitted on the training sample for a given cross-fitting fold.
The following conditions are assumed to hold uniformly across the $K$ folds:
\begin{enumerate}[label=(\alph*)]
    \item
    $\sup_{\bX\in\mathcal X}|\hat g(\bX)-g_0(\bX)| + \sup_{\bX\in\mathcal X}|\hat\pi(\bX)-\pi_0(\bX)| =o_p(1).$

    \item
    $\|\hat g(\bX)-g_0(\bX)\|_{P,2} + \|\hat\pi(\bX)-\pi_0(\bX)\|_{P,2} =o_p(n^{-1/4}).$
\end{enumerate}
\end{assumption}
Part (a) requires the nuisance estimators to be uniformly consistent over the covariate space $\mathcal{X}$-- i.e., the maximum difference between the estimated and true functions will vanish in probability as the sample size increases. 
Part (b) specifies the required convergence rate for the nuisance estimators in terms of the $L_2$ norm.
Specifically, the sum of the $L_2$ errors is $o_p(n^{-1/4})$, a weaker requirement that allows nuisance estimation to converge at rates slower than the root-$n$ rate.
Machine learners such as LASSO regression, random forests, and neural networks can achieve this rate in common settings \citep{diaz20MLRate, farrell21deep}. 

From the two assumptions above, we can obtain the following remarks: 
\begin{remark} \label{remark:nuisance}
Based on Assumption 1(c) and Assumption 2, the nuisance estimators $\hat{g}(\bX)$ and $\hat{\pi}(\bX)$ are uniformly bounded on $\mathcal{X}$. 
Specifically, we have that $\hat{g}(\bX)$ is bounded by a finite constant and $\hat{\pi}(\bX)$ is bounded away from 0 and 1. 
\end{remark}

\begin{remark} \label{remark:mu}
The nuisance parameter $\mu(\bX_i)$, as defined in equation~(\ref{eq:mu}), is a smooth function with respect to $\pi(\bX_i)$ and $\theta$.
Since $\theta$ lies in a compact set and $\pi(\bX)$ is bounded away from 0 and 1 by Assumption~\ref{ass:reg}(a) and~(c), $\hat{\mu}(\bX_i)$ inherits the same uniform consistency and $L_2$ convergence rate properties as $\hat{\pi}(\bX_i)$.
\end{remark}

We follow the definition of Neyman orthogonality in \cite{chernozhukov18double} and \cite{chiang22multiway}.
The following proposition establishes that our proposed estimating equation is Neyman orthogonal.
A formal proof is provided in Appendix~\ref{sec:proofOrthogonal}.

\begin{proposition}[Neyman Orthogonality of the Proposed Estimating Equation]
\label{prop:orthogonal}
Our proposed score function in equation~(\ref{eq:score}) obeys the Neyman orthogonality condition at $\eta_0=(g_0, \mu_0)$, namely
$$
\partial_\eta\,\EE(\psi(W_i; \theta_0,\eta_0))[\eta-\eta_0]=0.
$$
\end{proposition}
Proposition~\ref{prop:orthogonal} is used with Assumptions~\ref{ass:reg} and~\ref{ass:ml} to verify the conditions required for our main result below.
These conditions are stated in Appendix~\ref{sec:conditions} and verified in Appendix~\ref{sec:proofNormal} to prove Theorem~\ref{thm:normal}, which establishes the asymptotic normality of the DML estimator for a partially linear Poisson model. 

\begin{theorem}[Asymptotic Normality of the Proposed DML Estimator]
\label{thm:normal}
Under Proposition~\ref{prop:orthogonal} and Assumptions~\ref{ass:reg} and~\ref{ass:ml}, $\hat{\theta}$ derived from our proposed estimating equation satisfies
\begin{align*}
    \sqrt{n}\sigma_0^{-1}(\hat{\theta}-\theta_0) 
    =-\frac{1}{\sqrt{n}}\sum_{i=1}^n (\sigma_0J_0)^{-1}\psi(W_i;\theta_0,\eta_0)
    +o_p(1) \xrightarrow{d}{N}(0,1),
\end{align*}
where $J_0$ and $\sigma_0^2$ are defined in equations~(\ref{eq:J0}) and~(\ref{eq:asyvar}), respectively.
\end{theorem}
This theorem supports valid statistical inference for our proposed DML estimator.
The proof of this theorem is detailed in Appendix~\ref{sec:proofNormal}.

\begin{corollary}[Consistency of the variance estimator]
\label{cor:variance}
Under the conditions of Theorem~\ref{thm:normal}, $\hat{\sigma}^2 \xrightarrow{p} \sigma_0^2$.
Consequently, $\sqrt n(\hat\theta-\theta_0)/\hat\sigma$ converges in distribution to a standard normal random variable.
\end{corollary}

The score function in equation~(\ref{eq:score}) also has the double-robustness property.
In the context of the partially linear Poisson model, the estimating equation remains unbiased at $\theta_0$ if either the treatment nuisance function is correctly specified, $\mu(\bX_i)=\mu_0(\bX_i)$, or the outcome nuisance function is correctly specified, $g(\bX_i)=g_0(\bX_i)$.
We formally establish this result in Appendix~\ref{sec:proofDR}.

\section{Simulation study} \label{sec:sim}
This section presents a simulation study evaluating the finite-sample performance of our proposed Poisson DML estimator and comparing it with alternative approaches.

\subsection{Data-generating process}
\label{sec:sim_dgp}
For each simulation replicate, we generated four mutually independent baseline covariates:
\begin{align*}
X_{ij}\overset{\mathrm{iid}}{\sim}\mathcal{N}(0,1),
\qquad j=1,\ldots,4.
\end{align*}
Here, the standard normal distribution was truncated to $[-10,10]$ to ensure finite first and second moments for the count outcome.
Treatment was generated according to
\begin{align*}
A_i\mid\bX_i
&\sim\mathrm{Bernoulli}\left(\pi_0(\bX_i)\right),
\end{align*}
where
\begin{align*}
    \pi_0(\bX_i) =\operatorname{expit}\big(0.5+0.8X_{i1}-0.5X_{i2}
    +0.3X_{i3}-0.6X_{i4}\big).
\end{align*}

Conditional on treatment and covariates, the outcome was generated as
\begin{align*}
Y_i\mid A_i,\bX_i
&\sim\mathrm{Poisson}\left(\lambda_s(A_i,\bX_i)\right) \qquad \text{and} \qquad
\lambda_s(A_i,\bX_i) =e^{\theta_0A_i+g_s(\bX_i)}
\end{align*}
where $g_s$ and $\lambda_s$ are the true underlying functions under scenario $s\in\{1,2\}$.
The outcome nuisance function under the simpler covariate structure was
\begin{align*}
    g_{1}(\bX_i)
    =0.20\,X_{i1}^{2}+0.05\,X_{i2}+0.05\,X_{i3} 
    -0.30\,X_{i2}X_{i3}+0.50\,\mathbf 1\{X_{i4}>2\}.
\end{align*}
The nuisance function under the more complex covariate structure $g_2$ was
\begin{align*}
    g_{2}(\bX_i)
    &= 0.20\,X_{i1}^{2} + 0.10\,X_{i2} + 0.20\,X_{i3}
    - 0.10\,X_{i2}X_{i3}
    + 0.30\,\mathbf 1\{X_{i4}>2\} \\
    &\quad + 0.15\,X_{i1}X_{i2} + 0.10\,X_{i3}^{2}
    + 0.20\,X_{i1}X_{i4} + 0.35\,\sin(X_{i2}) \\
    &\quad + 0.20\,\log (|X_{i3}|+1)
    + 0.25\,\mathbf 1\{X_{i1}>2\}\,\mathbf 1\{X_{i2}<-1\}.
\end{align*}
Under these nuisance functions, the bounded covariate support ensures bounded baseline log-means, propensity scores bounded away from zero and one, and finite outcome moments of every order.

We considered $\theta_0\in\{0.5, 1.0, 2.0 \}$, three increasing treatment effect sizes corresponding to mean ratios of approximately 1.65, 2.72, and 7.39.
We used sample sizes of $1,000$ and $5,000$ and generated $N=500$ Monte Carlo replicates for each combination of nuisance scenario, sample size, and treatment effect size.

\subsection{Applying the proposed methods to simulated datasets} \label{sec:sim_methods}

We evaluated our proposed DML estimator with propensity scores estimated using either logistic regression or random forests (RF). 
Both methods used the partially linear Poisson neural networks described in Section~\ref{subsec:nuest} to estimate the outcome nuisance function $g_0(\bX_i)$. 
We used two-fold and five-fold cross-fitting to obtain out-of-fold predictions of nuisance functions.
The DML estimate of $\theta_0$ was obtained by solving the Neyman-orthogonal estimating equation.
Standard errors and Wald 95\% confidence intervals were calculated using the variance estimator described in Section~\ref{subsec:wald}.
We also included a DML estimator in which the true $g_0(\bX_i)$ from the data-generating process was substituted in place of the fitted nuisance $\hat g$ from the neural networks, while the propensity score remained estimated using five-fold cross-fitting.

\begin{figure}[!t]
    \centering
    \includegraphics[width=0.85\columnwidth]{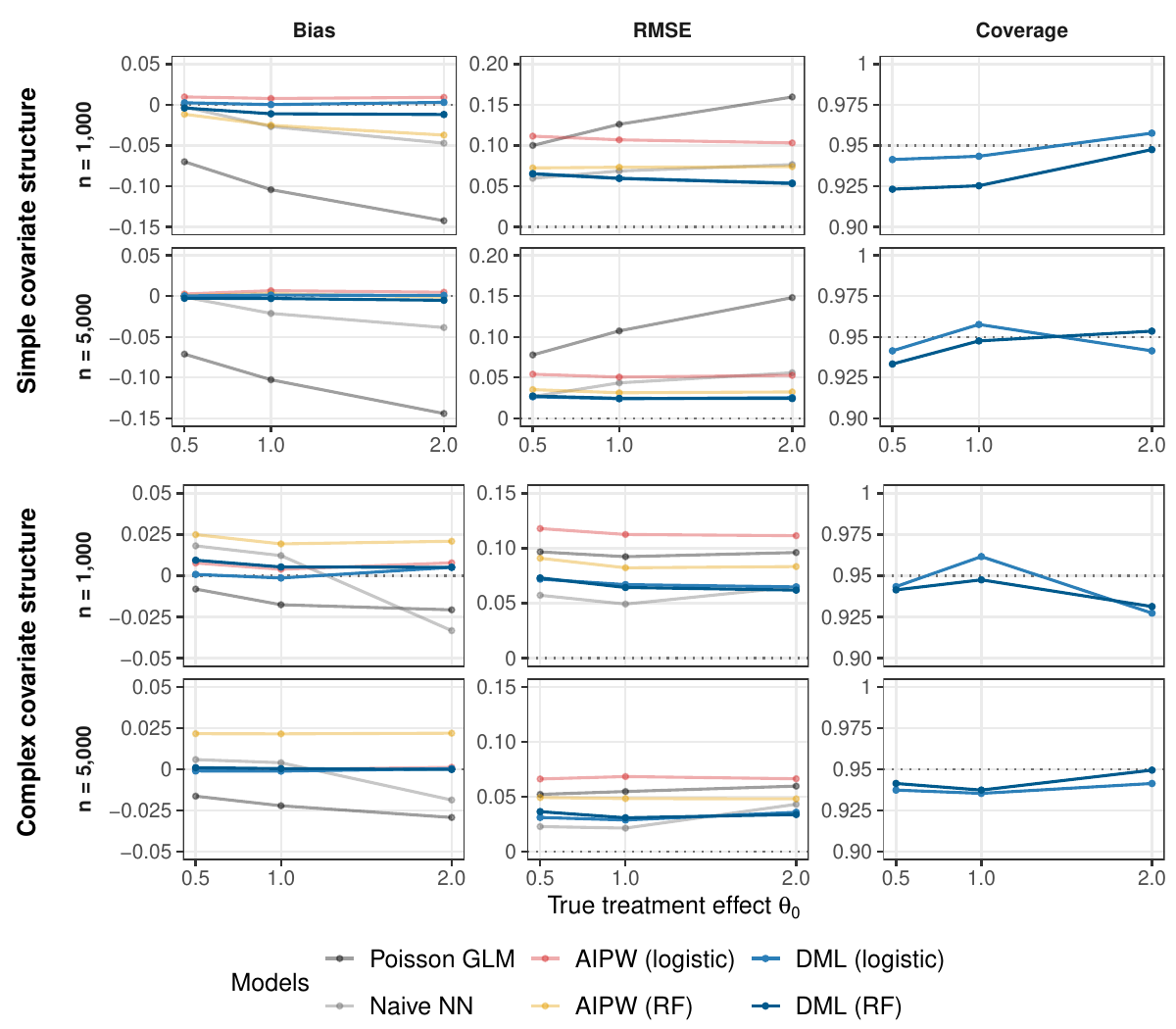}
    \caption{\textbf{Finite-sample performance of the DML estimator and comparative methods}. Bias and root mean square error (RMSE) are reported on the scale of $\theta$, and coverage denotes the empirical coverage probability of the 95\% Wald confidence interval. The DML estimators use five-fold cross-fitting.}
    \label{fig:sim_results}
\end{figure}

We used five-fold cross-validation to select the hyperparameters for the neural network and random forest models.
The first five of the 500 simulated datasets were reserved for hyperparameter tuning, and the configuration with the lowest average loss was applied to the remaining datasets.
The tuning grids are reported in Appendix~\ref{sec:sim_appendix}.
We summarized the performance of estimators using empirical bias and root mean square error (RMSE) and, for the DML estimators, empirical coverage probability (CP) of the 95\% Wald confidence intervals.

Comparative methods included a linearly adjusted Poisson GLM, a naive partially linear Poisson neural network (naive NN) estimator fitted using the full sample, and common AIPW estimators using a Poisson GLM outcome regression and either logistic regression or a random forest as the propensity score model.
For the AIPW estimator, we calculated the two marginal counterfactual mean counts and their ratio.
Under the partially linear Poisson model, this ratio equals $\exp\{\theta_0\}$ and its logarithm estimates $\theta_0$, so AIPW targets the same estimand as the other methods.
For the naive NN estimator, the neural network was fitted to the full sample of each dataset using the hyperparameters selected during tuning. 
The fitted coefficient of the treatment node was reported directly as the estimate of $\theta_0$, so this estimator served as a comparator that did not use cross-fitting or the Neyman-orthogonal estimating equation.
The implementation of these comparative methods can be found in Appendix~\ref{sec:sim_appendix}.

\subsection{Results}

Figure~\ref{fig:sim_results} summarizes the empirical bias, RMSE, and CP of the 95\% confidence intervals under the two covariate structures for the DML estimators using five-fold cross-fitting and other comparative methods.
Complete numerical results, including the DML estimators using two-fold cross-fitting and the true outcome nuisance, are reported in Tables~\ref{tab:sim_g1} and~\ref{tab:sim_g2}.

\begin{table}[!t]
\centering
\caption{Simulation results under the simple covariate structure $g_1(\bX_i)$.}
\label{tab:sim_g1}
\begingroup
\footnotesize
\setlength{\tabcolsep}{1.5pt}
\renewcommand{\arraystretch}{1.08}
\resizebox{\textwidth}{!}{%
\begin{tabular}{@{}l *{18}{r}@{}}
\toprule
& \multicolumn{6}{c}{$\theta_0=0.5$} & \multicolumn{6}{c}{$\theta_0=1.0$} & \multicolumn{6}{c}{$\theta_0=2.0$} \\
\cmidrule(lr){2-7}\cmidrule(lr){8-13}\cmidrule(lr){14-19}
& \multicolumn{3}{c}{$n=1{,}000$} & \multicolumn{3}{c}{$n=5{,}000$} & \multicolumn{3}{c}{$n=1{,}000$} & \multicolumn{3}{c}{$n=5{,}000$} & \multicolumn{3}{c}{$n=1{,}000$} & \multicolumn{3}{c}{$n=5{,}000$} \\
\cmidrule(lr){2-4}\cmidrule(lr){5-7}\cmidrule(lr){8-10}\cmidrule(lr){11-13}\cmidrule(lr){14-16}\cmidrule(lr){17-19}
\textbf{Method} & \textbf{Bias} & \textbf{RMSE} & \textbf{CP} & \textbf{Bias} & \textbf{RMSE} & \textbf{CP} & \textbf{Bias} & \textbf{RMSE} & \textbf{CP} & \textbf{Bias} & \textbf{RMSE} & \textbf{CP} & \textbf{Bias} & \textbf{RMSE} & \textbf{CP} & \textbf{Bias} & \textbf{RMSE} & \textbf{CP} \\
\midrule
\multicolumn{19}{l}{\textbf{Methods for comparison}} \\
\midrule
Poisson GLM & $-0.0699$ & 0.1001 & \textemdash & $-0.0712$ & 0.0778 & \textemdash & $-0.1042$ & 0.1261 & \textemdash & $-0.1027$ & 0.1073 & \textemdash & $-0.1424$ & 0.1596 & \textemdash & $-0.1442$ & 0.1482 & \textemdash \\
Naive NN & $-0.0025$ & 0.0597 & \textemdash & $-0.0011$ & 0.0266 & \textemdash & $-0.0267$ & 0.0687 & \textemdash & $-0.0212$ & 0.0435 & \textemdash & $-0.0471$ & 0.0766 & \textemdash & $-0.0386$ & 0.0559 & \textemdash \\
AIPW (GLM + logistic) & $\phantom{-}0.0097$ & 0.1115 & \textemdash & $\phantom{-}0.0026$ & 0.0541 & \textemdash & $\phantom{-}0.0078$ & 0.1070 & \textemdash & $\phantom{-}0.0067$ & 0.0506 & \textemdash & $\phantom{-}0.0092$ & 0.1032 & \textemdash & $\phantom{-}0.0048$ & 0.0527 & \textemdash \\
AIPW (GLM + RF) & $-0.0116$ & 0.0725 & \textemdash & $\phantom{-}0.0008$ & 0.0352 & \textemdash & $-0.0251$ & 0.0733 & \textemdash & $\phantom{-}0.0030$ & 0.0313 & \textemdash & $-0.0372$ & 0.0741 & \textemdash & $-0.0012$ & 0.0323 & \textemdash \\
\midrule
\multicolumn{19}{l}{\textbf{DML with logistic propensity scores}} \\
\midrule
DML ($K\!=\!2$) & $\phantom{-}0.0037$ & 0.0734 & 0.9273 & $\phantom{-}0.0002$ & 0.0270 & 0.9414 & $-0.0012$ & 0.0642 & 0.9313 & $\phantom{-}0.0017$ & 0.0247 & 0.9515 & $\phantom{-}0.0042$ & 0.0583 & 0.9414 & $\phantom{-}0.0013$ & 0.0258 & 0.9455 \\
DML ($K\!=\!5$) & $\phantom{-}0.0025$ & 0.0655 & 0.9414 & $\phantom{-}0.0003$ & 0.0264 & 0.9414 & $\phantom{-}0.0003$ & 0.0595 & 0.9434 & $\phantom{-}0.0013$ & 0.0240 & 0.9576 & $\phantom{-}0.0030$ & 0.0539 & 0.9576 & $\phantom{-}0.0009$ & 0.0253 & 0.9414 \\
DML (true $g$) & $\phantom{-}0.0004$ & 0.0592 & 0.9434 & $\phantom{-}0.0001$ & 0.0257 & 0.9414 & $-0.0021$ & 0.0518 & 0.9495 & $\phantom{-}0.0013$ & 0.0229 & 0.9475 & $\phantom{-}0.0010$ & 0.0469 & 0.9556 & $\phantom{-}0.0010$ & 0.0217 & 0.9475 \\
\midrule
\multicolumn{19}{l}{\textbf{DML with random forest propensity scores}} \\
\midrule
DML ($K\!=\!2$) & $-0.0031$ & 0.0678 & 0.9394 & $-0.0028$ & 0.0285 & 0.9232 & $-0.0117$ & 0.0625 & 0.9253 & $-0.0038$ & 0.0256 & 0.9434 & $-0.0139$ & 0.0582 & 0.9293 & $-0.0060$ & 0.0246 & 0.9434 \\
DML ($K\!=\!5$) & $-0.0039$ & 0.0652 & 0.9232 & $-0.0026$ & 0.0273 & 0.9333 & $-0.0110$ & 0.0598 & 0.9253 & $-0.0029$ & 0.0244 & 0.9475 & $-0.0119$ & 0.0533 & 0.9475 & $-0.0051$ & 0.0242 & 0.9535 \\
DML (true $g$) & $\phantom{-}0.0006$ & 0.0569 & 0.9374 & $\phantom{-}0.0002$ & 0.0255 & 0.9434 & $-0.0015$ & 0.0508 & 0.9596 & $\phantom{-}0.0015$ & 0.0230 & 0.9495 & $\phantom{-}0.0020$ & 0.0483 & 0.9535 & $\phantom{-}0.0010$ & 0.0217 & 0.9556 \\
\bottomrule
\end{tabular}%
}
\endgroup
\par\smallskip
\parbox{\textwidth}{\footnotesize
Bias and RMSE are on the scale of $\theta_0$; CP is the empirical coverage probability of the nominal 95\% Wald confidence interval.
DML (true $g$) uses the true outcome nuisance from the data-generating process and five-fold cross-fitting to estimate propensity scores.
}
\end{table}

\begin{table}[!t]
\centering
\caption{Simulation results under the complex covariate structure $g_2(\bX_i)$.}
\label{tab:sim_g2}
\begingroup
\footnotesize
\setlength{\tabcolsep}{1.5pt}
\renewcommand{\arraystretch}{1.08}
\resizebox{\textwidth}{!}{%
\begin{tabular}{@{}l *{18}{r}@{}}
\toprule
& \multicolumn{6}{c}{$\theta_0=0.5$} & \multicolumn{6}{c}{$\theta_0=1.0$} & \multicolumn{6}{c}{$\theta_0=2.0$} \\
\cmidrule(lr){2-7}\cmidrule(lr){8-13}\cmidrule(lr){14-19}
& \multicolumn{3}{c}{$n=1{,}000$} & \multicolumn{3}{c}{$n=5{,}000$} & \multicolumn{3}{c}{$n=1{,}000$} & \multicolumn{3}{c}{$n=5{,}000$} & \multicolumn{3}{c}{$n=1{,}000$} & \multicolumn{3}{c}{$n=5{,}000$} \\
\cmidrule(lr){2-4}\cmidrule(lr){5-7}\cmidrule(lr){8-10}\cmidrule(lr){11-13}\cmidrule(lr){14-16}\cmidrule(lr){17-19}
\textbf{Method} & \textbf{Bias} & \textbf{RMSE} & \textbf{CP} & \textbf{Bias} & \textbf{RMSE} & \textbf{CP} & \textbf{Bias} & \textbf{RMSE} & \textbf{CP} & \textbf{Bias} & \textbf{RMSE} & \textbf{CP} & \textbf{Bias} & \textbf{RMSE} & \textbf{CP} & \textbf{Bias} & \textbf{RMSE} & \textbf{CP} \\
\midrule
\multicolumn{19}{l}{\textbf{Methods for comparison}} \\
\midrule
Poisson GLM & $-0.0081$ & 0.0967 & \textemdash & $-0.0163$ & 0.0522 & \textemdash & $-0.0176$ & 0.0923 & \textemdash & $-0.0221$ & 0.0548 & \textemdash & $-0.0207$ & 0.0961 & \textemdash & $-0.0292$ & 0.0597 & \textemdash \\
Naive NN & $\phantom{-}0.0181$ & 0.0572 & \textemdash & $\phantom{-}0.0059$ & 0.0229 & \textemdash & $\phantom{-}0.0122$ & 0.0492 & \textemdash & $\phantom{-}0.0040$ & 0.0215 & \textemdash & $-0.0333$ & 0.0649 & \textemdash & $-0.0186$ & 0.0432 & \textemdash \\
AIPW (GLM + logistic) & $\phantom{-}0.0077$ & 0.1179 & \textemdash & $-0.0003$ & 0.0662 & \textemdash & $\phantom{-}0.0039$ & 0.1125 & \textemdash & $-0.0001$ & 0.0684 & \textemdash & $\phantom{-}0.0078$ & 0.1114 & \textemdash & $\phantom{-}0.0012$ & 0.0664 & \textemdash \\
AIPW (GLM + RF) & $\phantom{-}0.0249$ & 0.0909 & \textemdash & $\phantom{-}0.0216$ & 0.0493 & \textemdash & $\phantom{-}0.0193$ & 0.0822 & \textemdash & $\phantom{-}0.0215$ & 0.0484 & \textemdash & $\phantom{-}0.0209$ & 0.0833 & \textemdash & $\phantom{-}0.0219$ & 0.0481 & \textemdash \\
\midrule
\multicolumn{19}{l}{\textbf{DML with logistic propensity scores}} \\
\midrule
DML ($K\!=\!2$) & $\phantom{-}0.0022$ & 0.0736 & 0.9273 & $-0.0003$ & 0.0325 & 0.9434 & $\phantom{-}0.0066$ & 0.1309 & 0.9475 & $-0.0005$ & 0.0319 & 0.9394 & $\phantom{-}0.0051$ & 0.0726 & 0.9333 & $-0.0003$ & 0.0374 & 0.9293 \\
DML ($K\!=\!5$) & $\phantom{-}0.0008$ & 0.0719 & 0.9434 & $-0.0010$ & 0.0312 & 0.9374 & $-0.0014$ & 0.0669 & 0.9616 & $-0.0011$ & 0.0289 & 0.9354 & $\phantom{-}0.0051$ & 0.0648 & 0.9273 & $\phantom{-}0.0004$ & 0.0359 & 0.9414 \\
DML (true $g$) & $-0.0001$ & 0.0492 & 0.9495 & $-0.0006$ & 0.0215 & 0.9475 & $-0.0028$ & 0.0428 & 0.9576 & $-0.0000$ & 0.0204 & 0.9394 & $\phantom{-}0.0001$ & 0.0408 & 0.9333 & $\phantom{-}0.0007$ & 0.0185 & 0.9455 \\
\midrule
\multicolumn{19}{l}{\textbf{DML with random forest propensity scores}} \\
\midrule
DML ($K\!=\!2$) & $\phantom{-}0.0109$ & 0.0811 & 0.9152 & $\phantom{-}0.0004$ & 0.0347 & 0.9253 & $\phantom{-}0.0060$ & 0.0731 & 0.9455 & $\phantom{-}0.0008$ & 0.0326 & 0.9414 & $\phantom{-}0.0077$ & 0.0755 & 0.9152 & $-0.0000$ & 0.0364 & 0.9374 \\
DML ($K\!=\!5$) & $\phantom{-}0.0094$ & 0.0730 & 0.9414 & $\phantom{-}0.0010$ & 0.0365 & 0.9414 & $\phantom{-}0.0053$ & 0.0643 & 0.9475 & $\phantom{-}0.0004$ & 0.0309 & 0.9374 & $\phantom{-}0.0051$ & 0.0619 & 0.9313 & $-0.0001$ & 0.0339 & 0.9495 \\
DML (true $g$) & $-0.0000$ & 0.0466 & 0.9535 & $-0.0007$ & 0.0212 & 0.9394 & $-0.0027$ & 0.0414 & 0.9657 & $-0.0002$ & 0.0200 & 0.9434 & $-0.0001$ & 0.0406 & 0.9333 & $\phantom{-}0.0007$ & 0.0181 & 0.9515 \\
\bottomrule
\end{tabular}%
}
\endgroup
\par\smallskip
\parbox{\textwidth}{\footnotesize
Bias and RMSE are on the scale of $\theta_0$; CP is the empirical coverage probability of the nominal 95\% Wald confidence interval.
DML (true $g$) uses the true outcome nuisance from the data-generating process and five-fold cross-fitting to estimate propensity scores.
}
\end{table}

Under the simple covariate structure $g_1$, the linearly adjusted Poisson GLM generally showed the largest absolute bias, and its bias increased as the treatment effect increased.
The naive NN estimator reduced the bias relative to the GLM but generally remained more biased than the DML estimators.
This suggests that orthogonalization and cross-fitting in DML helped reduce the remaining bias associated with machine-learning nuisance estimation.
Compared with the AIPW estimators, the DML estimators using five-fold cross-fitting had lower RMSE in every setting.
This pattern held for both logistic regression and random forest propensity score models.
In general, the DML estimators had small bias, RMSE that decreased as the sample size increased, and coverage close to the nominal level.

Under the more complex covariate structure $g_2$, the five-fold DML estimators remained nearly unbiased and had lower RMSE than both AIPW estimators in every setting, while maintaining coverage close to the nominal level.
Within DML, propensity models using logistic regression and random forest had similar results.
Differences between two-fold and five-fold cross-fitting were generally modest, although the latter often produced marginally lower RMSE.
Notably, DML using the true outcome nuisance generally achieved the lowest RMSE and small bias (Tables~\ref{tab:sim_g1} and~\ref{tab:sim_g2}), suggesting the importance of reliable estimation of nuisance functions even when an orthogonal estimating equation is used.

\section{Empirical study} \label{sec:emp}

In this section, we apply the proposed DML estimator to the synthetic CCASAnet data and compare the estimated treatment effect with those from the Poisson GLM and AIPW estimators.

\subsection{Data description}

We applied our method to a synthetic dataset from the Caribbean, Central and South America Network for HIV Research (CCASAnet), a collaboration that harmonizes clinical HIV cohort data across sites in Latin America and the Caribbean.
The cohort consisted of 49,606 adults living with HIV from six countries as of December 2024 and spanned more than 30 years of follow-up.
The synthetic data closely reproduced longitudinal distributions, missingness patterns, time-to-event estimates, and risk-factor relationships in the original CCASAnet cohort without disclosing sensitive patient-level records \citep{liang26HIVsyn}.

The treatment of interest was dolutegravir-based antiretroviral therapy (ART). Treatment was classified as DTG-based if the initial regimens included DTG and as non-DTG otherwise.
The target population comprised ART-naive patients whose first recorded ART had a documented start date on or after January 1, 2013, the year when DTG was approved by the FDA.
After excluding patients with incomplete baseline covariates, the analytic sample consisted of 7,141 patients: 2,223 initiated a DTG-based regimen and 4,918 initiated a non-DTG regimen.

The primary outcome was the number of ADEs occurring within five years after ART initiation.
We retained patients with less than five years of recorded follow-up, with observed follow-up duration as an adjusted covariate in all models to account for the difference in observation time.
Finally, our \textbf{target causal estimand} is the conditional effect of DTG-based ART versus non-DTG ART on the count of the ADEs within five years.
Specifically, it is the ratio of the expected five-year ADE counts under the two treatment regimens, conditional on covariates.

The treatment and outcome models used the same set of baseline characteristics and demographic information, including baseline CD4 count, age, school level, sex, country/clinical center, HIV acquisition mode, and time of ART initiation. 
Descriptive statistics for these baseline covariates are reported in Appendix~\ref{sec:emp_appendix}, Table~\ref{tab:emp_baseline}.

\subsection{Model fitting and hyperparameter tuning}

We compared our DML estimator with the Poisson GLM and AIPW models.
We followed an implementation similar to that used in the simulation study in Section~\ref{sec:sim_methods}.

We used five-fold cross-fitting for all DML nuisance models. 
For each of the five outer cross-fitting folds, we tuned the hyperparameters within the outer training set using an 80\%/20\% inner training-validation split for the neural network and random forest models.
Using the selected hyperparameters, we refitted the models on the complete outer training set and used them to predict nuisance functions for the held-out outer fold.

For AIPW, the outcome model was estimated using a full-sample Poisson GLM, and the propensity score was estimated using five-fold cross-fitted logistic regression or random forests.
Bootstrap was used to compute the 95\% confidence intervals. 
For AIPW with logistic regression, the estimated propensity scores were truncated to $[0.01, 0.99]$ to mitigate the influence of extreme propensity scores.

\subsection{Results}

Table~\ref{tab:empirical} presents the estimated treatment effects on the log scale and their exponentiated values, interpreted as conditional five-year expected count ratios.

\begin{table}[t]
\caption{\textbf{Estimated DTG-Based Treatment Effect on Five-Year AIDS-Defining Event Counts:} Logistic regression (logistic) and random forest (RF) are used for the propensity score models. Point estimates and 95\% confidence intervals are reported for $\theta$ and $\exp\{\theta\}$, with the latter being the conditional five-year expected count ratio comparing DTG and non-DTG treatment. Wald confidence intervals are reported for Poisson GLM and DML; non-parametric bootstrap quantile intervals are reported for AIPW. }
\label{tab:empirical}
\begin{center}
\small
\setlength{\tabcolsep}{5pt}
\setlength{\heavyrulewidth}{1pt}
\setlength{\lightrulewidth}{0.6pt}
\renewcommand{\arraystretch}{1.12}
\begin{tabular*}{\textwidth}{@{\extracolsep{\fill}}lccccc@{}}
\toprule
Method
& $\hat{\theta}$
& SE
& 95\% CI for $\theta$
& $\exp\{\hat{\theta}\}$
& 95\% CI for $\exp\{\theta\}$ \\
\midrule
Poisson GLM
& $0.354$ & $0.162$ & $[0.037,\,0.671]$
& $1.425$ & $[1.038,\,1.957]$ \\
AIPW (logistic)
& $0.126$ & $0.413$ & $[-0.857,\,0.790]$
& $1.134$ & $[0.424,\,2.204]$ \\
AIPW (RF)
& $0.326$ & $0.170$ & $[-0.010,\,0.658]$
& $1.385$ & $[0.990,\,1.931]$ \\
DML (logistic)
& $0.305$ & $0.158$ & $[-0.004,\,0.614]$
& $1.357$ & $[0.996,\,1.848]$ \\
DML (RF)
& $0.258$ & $0.155$ & $[-0.046,\,0.561]$
& $1.294$ & $[0.955,\,1.752]$ \\
\bottomrule
\end{tabular*}
\end{center}
\end{table}

The two DML estimates, 1.357 (95\% CI: 0.996, 1.848) with the logistic propensity score model and 1.294 (95\% CI: 0.955, 1.752) with the RF propensity score model, were similar to each other.
These estimates corresponded to 35.7\% and 29.4\% higher expected five-year ADE counts for DTG-based regimens than for non-DTG regimens.
However, both confidence intervals included the null value of 1.
Compared with the Poisson GLM estimate of 1.425 (95\% CI: 1.038, 1.957), the point estimates of DML methods were closer to the null, while their confidence intervals suggested insufficient evidence that DTG was associated with higher ADE counts than non-DTG regimens, contrary to the Poisson GLM results.

The DML estimates had smaller estimated standard errors than the AIPW estimates for both propensity score models.
The precision gain was particularly large for the estimates using the logistic propensity model, as the estimated standard error decreased from 0.413 for AIPW to 0.158 for DML.
This difference may suggest that the DML estimator is more stable in the presence of extreme propensity scores, because DML does not invert the propensity scores to construct inverse probability weights.
The distributions of the raw and truncated inverse probability weights from the logistic regression and the random forest models are presented in Appendix~\ref{sec:emp_weight_distribution}.

\section{Conclusion}

We presented our DML estimator for the treatment effect coefficient in a partially linear Poisson model.
Using simulations, we demonstrated that DML estimators had smaller bias and decreasing RMSE in most settings.
More importantly, coverage was near the nominal 95\% level, supporting the validity of statistical inference for our estimator.
In the synthetic CCASAnet application, DML found insufficient evidence that DTG was associated with higher five-year ADE counts, in contrast to the Poisson GLM. 
DML also provided more precise estimates across propensity score models and suggested greater robustness to extreme propensity scores than AIPW.

Some limitations may affect the interpretation and generalizability of the CCASAnet analysis.
First, loss to follow-up may bias the estimated treatment effect if associated with patient characteristics or evolving health status. 
As a result, inverse probability of censoring weighting could be incorporated to address this potential bias.
Second, although our model assumes a conditional Poisson distribution, the unbiasedness and Neyman orthogonality of the proposed score function depend only on the conditional mean, and do not require the mean to equal its variance. 
In addition, the empirical sandwich variance estimator allows misspecification of the conditional variance while still providing valid variance estimation asymptotically.
However, overdispersion may reduce its efficiency relative to estimators that appropriately incorporate the correct conditional variance.
Future work could develop an orthogonal estimating equation based on negative binomial regression or other flexible count models.
Third, the partially linear model assumes a constant conditional count mean ratio across covariate values. 
This restriction provides a parsimonious model for treatment effect estimation, but does not allow treatment heterogeneity. 
Future work could incorporate prespecified treatment-effect modification by allowing a subset of the covariates to interact with the treatment.

\bibliographystyle{apalike}
\bibliography{dml_poisson}

\clearpage
\appendix
\section*{Appendix}

\section{Construction of the Neyman-orthogonal score}
\label{sec:scoreConstruction}
This section shows the construction of the Neyman-orthogonal score in equation~(\ref{eq:score}), where we follow the concentrating-out approach for infinite-dimensional nuisance parameters in Section 2.2.3 of \cite{chernozhukov18double}.

Let the log-likelihood from the partially linear Poisson model be the criterion function
\begin{align*}
    \ell\left(W_i;\theta,g\right)=Y_i\left(\theta A_i+g\left(\bX_i\right)\right)-e^{\theta A_i+g\left(\bX_i\right)},
\end{align*}
where $\theta$ and $g$ are the target and nuisance parameters taking values in parameter spaces $\Theta$ and $G$, respectively. Define the concentrated-out nonparametric component by
\begin{align*}
    g_\theta\left(\bX_i\right)=\underset{g\in G}{\operatorname{arg\,max}}\;\EE\left[\ell\left(W_i;\theta,g\right)\mid\bX_i\right].
\end{align*}
For a fixed $\bX_i$, the following $g_\theta$ maximizes $\EE\left[\ell\left(W_i;\theta,g\right)\mid\bX_i\right]$:
\begin{align*}
    g_\theta\left(\bX_i\right)=\log\EE\left[Y_i\mid\bX_i\right]-\log\EE\left[e^{\theta A_i}\mid\bX_i\right].
\end{align*}
Its partial derivative with respect to $\theta$ is
\begin{align*}
    \frac{\partial g_\theta\left(\bX_i\right)}{\partial\theta}=-\frac{\EE\left[A_i e^{\theta A_i}\mid\bX_i\right]}{\EE\left[e^{\theta A_i}\mid\bX_i\right]}:=-\mu_\theta\left(\bX_i\right).
\end{align*}

Now consider the following score function:
\begin{align*}
\psi\left(W_i;\theta,\eta\right)=\frac{d\ell\left(W_i;\theta,\eta\left(\theta\right)\right)}{d\theta},
\end{align*}
where $d/d\theta$ denotes the full derivative with respect to $\theta$, including differentiation through the nuisance path $\eta\,(\theta)$. For our model, this becomes
\begin{align*}
    \frac{d\ell\left(W_i;\theta,g_\theta\right)}{d\theta}
    &=\partial_\theta\ell\left(W_i;\theta,g_\theta\right)
    +\partial_g\ell\left(W_i;\theta,g_\theta\right)
    \frac{\partial g_\theta\left(\bX_i\right)}{\partial\theta}.
\end{align*}
Using the Poisson log-likelihood as the criterion function gives
\begin{align*}
    \frac{d\ell\left(W_i;\theta,g_\theta\right)}{d\theta}
    &=A_i\left(Y_i-e^{\theta A_i+g_\theta\left(\bX_i\right)}\right)
    -\mu_\theta\left(\bX_i\right)
    \left(Y_i-e^{\theta A_i+g_\theta\left(\bX_i\right)}\right)\\
    &= \left(A_i-\mu_\theta\left(\bX_i\right)\right)\left(Y_i-e^{\theta A_i+g_\theta\left(\bX_i\right)}\right),
\end{align*}
which is the score in equation~(\ref{eq:score}).
The orthogonality of this score is formally proved in Proposition~\ref{prop:orthogonal}.

\section{DML estimation and inference algorithm} \label{sec:dmlalgorithm}

Algorithm~\ref{alg:poissondml} summarizes the nuisance estimation, root-finding, and parameter inference algorithm for the DML procedure described in Section~\ref{sec:procedure}.
The held-out fold is used only to construct out-of-sample nuisance predictions.

\begin{algorithm}[H]
\caption{Cross-fitted Poisson DML estimation and inference}
\label{alg:poissondml}
\small
\begin{algorithmic}[1]
\State \textbf{Input:} Data $W=\{(Y_i,A_i,\bX_i)\}_{i=1}^n$; number of folds $K$; tolerance $\delta$; maximum iterations $M$; significance level $\alpha$.

\Statex \vspace{0.5ex}
\Statex \textbf{Part 1: Cross-fitting for nuisance estimation}
\State Randomly partition $W$ into $K$ nonoverlapping held-out samples
$W_{\mathrm{hold}}^{(1)},\ldots,W_{\mathrm{hold}}^{(K)}$.
\State Initialize $\hat g=(0,\ldots,0)$, $\hat\pi=(0,\ldots,0)$, and an empty list $L$.
\For{$k=1,\ldots,K$}
    \State Set $W_{\mathrm{hold}}\gets W_{\mathrm{hold}}^{(k)}$ and
    $W_{\mathrm{train}}\gets W \setminus W_{\mathrm{hold}}$.
    \State Tune the outcome model and propensity score model using only $W_{\mathrm{train}}$, when applicable.
    \State Fit the partially linear Poisson outcome model on $W_{\mathrm{train}}$ to obtain $\tilde\theta_k$ and $\hat g^{(k)}$.
    \State Append $\tilde\theta_k$ to $L$ and use $\hat g^{(k)}$ to predict
    $\hat g_i$ for every observation in $W_{\mathrm{hold}}$.
    \State Fit the propensity score model on $W_{\mathrm{train}}$ to obtain $\hat\pi^{(k)}(\bX)=\widehat{\PP}(A=1\mid\bX)$.
    \State Use $\hat\pi^{(k)}$ to predict and truncate $\hat\pi_i$ for every observation in
    $W_{\mathrm{hold}}$.
\EndFor

\Statex \vspace{0.5ex}
\Statex \textbf{Part 2: Root-finding for the estimating equation}
\State Initialize $\theta^{(0)}\gets K^{-1}\sum_{k=1}^K\tilde\theta_k$ and $m\gets0$.
\Repeat
    \State For every $i$, compute
    \Statex \hspace{2em}$\displaystyle
        \hat\lambda_i^{(m)}=e^{\theta^{(m)}A_i+\hat g_i},\qquad
        \hat\mu_i^{(m)}=\frac{\hat\pi_i e^{\theta^{(m)}}}{1-\hat\pi_i+\hat\pi_i e^{\theta^{(m)}}},\qquad
        \hat\psi_i^{(m)}=(A_i-\hat\mu_i^{(m)})(Y_i-\hat\lambda_i^{(m)}).$
    \State Compute the empirical score $\displaystyle G_n^{(m)}=n^{-1}\sum_{i=1}^n\hat\psi_i^{(m)}$.
    \State Compute its derivative
    \Statex \hspace{2em}$\displaystyle
        \hat J^{(m)}\gets\frac{1}{n}\sum_{i=1}^n
        \left[-\hat\mu_i^{(m)}(1-\hat\mu_i^{(m)})(Y_i-\hat\lambda_i^{(m)})
        -(A_i-\hat\mu_i^{(m)})A_i\hat\lambda_i^{(m)}\right].$
    \Statex \vspace{1mm} Update $\theta^{(m+1)}\gets\theta^{(m)}- G_n^{(m)}/\hat J^{(m)}$ and $m\gets m+1$.

\Until{$\displaystyle |\theta^{(m)}-\theta^{(m-1)}|<\delta$ or $m=M$}

\Statex \vspace{0.5ex}
\Statex \textbf{Part 3: Statistical inference using a Wald confidence interval}
\State Set $\hat\theta\gets\theta^{(m)}$ and evaluate $\hat\psi_i=\hat\psi_i(\hat\theta)$ and $\hat J=\hat J(\hat\theta)$.
\State Compute $\displaystyle \hat\sigma^2\gets\hat J^{-2}n^{-1}\sum_{i=1}^n\hat\psi_i^2$ and $\displaystyle \widehat{\operatorname{se}}(\hat\theta)\gets\sqrt{\hat\sigma^2/n}$.
\State Construct $\displaystyle \mathrm{CI}_{1-\alpha}(\theta_0)\gets\hat\theta\pm z_{1-\alpha/2}\widehat{\operatorname{se}}(\hat\theta)$.
\State \textbf{Output:} $\hat\theta$, $\exp\{\hat\theta\}$, $\widehat{\operatorname{se}}(\hat\theta)$, and the Wald confidence interval (with exponentiated endpoints for $\exp\{\theta_0\}$).
\end{algorithmic}
\end{algorithm}

\section{Proof of Neyman orthogonality for the estimating equation}
\label{sec:proofOrthogonal}

We use the definition of Neyman orthogonality in Equation 2.3 of \cite{chernozhukov18double}.
To prove that our estimating equation in equation~(\ref{eq:score}) is Neyman orthogonal, we need to verify the following conditions:

\begin{enumerate}[label=(\alph*)]
    \item $\EE[\psi(W_i; \theta_0,\eta_0)]=0$.
    \item The pathwise derivative map $\text{D}_r[\eta-\eta_0]$ defined in \cite{chernozhukov18double} exists for all $r\in [0, 1)$ and $\eta\in \mathcal T_N$ and vanishes at $r=0$; 
    namely, 
    $$
    \partial_\eta\,\EE[\psi(W_i; \theta_0,\eta_0)][\eta-\eta_0]=0.
    $$
\end{enumerate}

\begin{proof}
\label{proof:orthogonal}
\begin{enumerate}[label=\textsc{Step \arabic*.},
                  leftmargin=4em,
                  labelsep=0.5em,
                  itemsep=5pt,
                  topsep=0.5\baselineskip]
\item The score function in our study is $\psi(W_i; \theta, \eta_0) = (A_i - \mu_0)(Y_i - \exp\{\theta A_i + g_0\}).$
Let $\lambda_0=\exp\{\theta_0A_i+g_0\}$. The expectation of this score function is
\begin{align*}
\EE(\psi(W_i; \theta_0, \eta_0)) &= \EE(\EE((A_i - \mu_0)(Y_i - \lambda_0)|\bX_i)) \\
&= \EE(\EE(AY - A_i\lambda_0 - \mu_0 Y_i + \mu_0 \lambda_0 | \bX_i)) \\
&= \EE(\EE(AY|\bX_i)-\EE(A_i\lambda_0|\bX_i)-\EE(\mu_0Y_i|\bX_i)+\EE(\mu_0\lambda_0|\bX_i)) \\
&= \EE(0)=0.
\end{align*}

The first equality uses the law of iterated expectation. The third equality holds because
\begin{align*}
\EE(AY | \bX_i) &= \EE(A_i \cdot \EE(Y_i | A_i, \bX_i) | \bX_i) = \EE(A_i \lambda_0 | \bX_i),\\
\EE(\mu_0 Y_i | \bX_i) &= \EE(\mu_0\cdot \EE(Y_i | A_i, \bX_i) | \bX_i) = \EE(\mu_0 \lambda_0 | \bX_i).
\end{align*}

\item We want to show that $\partial_\eta\,\EE(\psi(W_i; \theta_0, \eta_0))[\eta - \eta_0]=0$.

Using the definition of Neyman orthogonality from \cite{chernozhukov18double}, we define $g_r=g_0+r(g-g_0)$ and $\mu_r=\mu_0+r(\mu-\mu_0)$. The following holds:

\begin{align} \label{eq:part.dev}
    &\partial_\eta\,\EE[\psi(W_i; \theta_0, \eta_0)]\left[\eta - \eta_0\right] \nonumber \\
    &\quad =\big.\partial_r \EE\Big(\psi(W_i;\theta_0,\eta_0+r(\eta-\eta_0)\Big)\big|_{r=0} \nonumber \\
    &\quad = \left.\partial_r\,\EE\left((A_i-\mu_r)(Y_i-e^{\theta_0A_i+g_r})\right)\right|_{r=0} \nonumber \\
    &\quad = \left.\EE\left(-(\mu-\mu_0)(Y_i-e^{\theta_0A_i+g_r})+(A_i-\mu_r)(-e^{\theta_0A_i+g_r})(g-g_0))\right)\right|_{r=0} \nonumber \\
    &\quad = \EE\left(-(\mu-\mu_0)(Y_i-e^{\theta_0A_i+g_0}) + (A_i-\mu_0)(-e^{\theta_0A_i+g_0})(g-g_0)\right) \nonumber \\
    &\quad = \underbrace{\EE\left(-(\mu-\mu_0)(Y_i-e^{\theta_0A_i+g_0})\right)}_{\text{Term }1} + \underbrace{\EE\left((A_i-\mu_0)(-e^{\theta_0A_i+g_0})(g-g_0)\right)}_{\text{Term }2}.
\end{align}

Since $\EE(Y_i\mid A_i,\bX_i) - \lambda_0 = 0$, Term 1 in equation~(\ref{eq:part.dev}) becomes
\begin{align*}
    \EE\left[-(\mu - \mu_0)(Y_i - e^{\theta_0 A_i + g_0})\right]
    &= \EE_{\bX_i}\Big[\EE_{A_i}\big(\EE_{Y_i}(-(\mu - \mu_0)(Y_i - \lambda_0)\mid A_i, \bX_i)\big)\Big] \\
    &= \EE_{\bX_i}\Big[\EE_{A_i}\big(-(\mu - \mu_0) \EE_{Y_i}(Y_i - \lambda_0\mid A_i, \bX_i)\big)\Big] = 0.
\end{align*}
The second equality uses the law of iterated expectation.

In Term 2 in equation~(\ref{eq:part.dev}), we have $\mu_0=\EE[A_i\lambda_0|\bX_i]/\,\EE[\lambda_0|\bX_i]$, so
\begin{align*}
\EE[(A_i - \mu_0)\lambda_0\mid \bX_i] &= \EE[A_i\lambda_0\mid \bX_i] - \EE[\mu_0\lambda_0\mid \bX_i] \\
&= \EE[A_i\lambda_0|\bX_i] - \frac{\EE[A_i\lambda_0|\bX_i]}{\EE[\lambda_0|\bX_i]} \cdot \EE[\lambda_0|\bX_i] = 0.
\end{align*}

Thus, Term 2 becomes
\begin{align*}
\EE[-(g - g_0)(A_i - \mu_0)e^{\theta_0A_i+g_0}] &= \EE_{\bX_i}[\EE_{A_i}[-(g - g_0)(A_i - \mu_0)\lambda_0\mid \bX_i]] \\
&= \EE_{\bX_i}[-(g - g_0) \underbrace{\EE_{A_i}[(A_i - \mu_0)\lambda_0\mid \bX_i]}_{=0}] = 0.
\end{align*}

Hence, we have proven Proposition \ref{prop:orthogonal} and shown that our score function obeys the Neyman-orthogonality condition.

\end{enumerate}
\end{proof}

\section{Conditions to verify for the DML asymptotic normality}
\label{sec:conditions}
Below, we state the conditions used to establish the asymptotic normality of the DML estimator based on the estimating equation in equation~(\ref{eq:score}).
These conditions correspond to Assumptions 3.3 and 3.4 in \cite{chernozhukov18double}.

\begin{condition}[Moment Condition with Neyman Orthogonality]
\label{cond:a1}
Let $c_0, c_1>0$ be finite constants. The following conditions hold:
\begin{enumerate}[label=(\alph*)]
    \item $\EE(\psi(W_i; \theta_0, \eta_0)) = 0$ and $\Theta$ contains an interval of length $c_1(n^{-1/2}\log n)$ centered at $\theta_0$.

    \item The map $(\theta, \eta) \mapsto \EE[\psi(W_i; \theta, \eta)]$ is twice continuously Gateaux-differentiable.

    \item $2|\EE\psi(W_i; \theta, \eta_0)|\geq \min\{|J_0(\theta-\theta_0)|, c_0\}$, where $\eta_0=\{g_0, \mu_0\}$, $0<c_0<\infty$, \\
    and $0<|J_0|=|\partial_\theta\EE\psi(W_i; \theta,\eta_0)\mid_{\theta=\theta_0}|<\infty$.

    \item The score function $\psi(W_i; \theta_0, \eta_0)$ obeys Neyman-orthogonality-- i.e.,
    \begin{equation*}
    \partial_{\eta}\EE\psi(W_i; \theta_0, \eta_0)[ \eta-\eta_0] = 0 \text{ for all } \eta \in \text{\Large$\varepsilon$},
    \end{equation*}
    where \text{\Large$\varepsilon$} is the parameter space of $\eta$.
\end{enumerate}
\end{condition}

\begin{condition}[Score Regularity and Quality of Estimation of Nuisance Parameters]
\label{cond:a2}
The following conditions hold:
\begin{enumerate}[label=(\alph*)]
    \item $\hat{\eta}$ belongs to the realization set $\mathcal{T}_n$ for each $k \in \{1, 2, \ldots, K\}$ with probability approaching 1, where $\mathcal{T}_n$ satisfies $\eta_0 \in \mathcal{T}_n$.

    \item The space of $\theta$, $\Theta$, is bounded, and for each $\eta \in \mathcal{T}_n$, the functional class $\mathcal{F}_{\eta} = \{\psi(\cdot; \theta, \eta): \theta \in \Theta\}$ is $P$-measurable and its uniform covering number satisfies the following condition: there exist positive constants $R$ and $\nu$ such that
    \begin{align*}
    \sup_{Q} \log \mathcal{N}\Big(
    \varepsilon \|F_{\eta}(\cdot)\|_{Q,2}, \mathcal{F}_{\eta}, \|\cdot\|_{Q,2}\Big)
    \leq \nu \log(R/\varepsilon),
    \quad \forall \varepsilon \in (0,1),
    \end{align*}
    where $F_{\eta}(\cdot)$ is a measurable envelope function of $\mathcal{F}_{\eta}$ such that $F_{\eta}(W_i) \geq |\psi(W_i;\theta,\eta)|$ for all $W_i$ and $\theta \in \Theta$, and there exists $q > 2$ such that $\|F_{\eta}\|_{P,q} = O(1)$.

    \item There exists $\tau_n$, a sequence of positive constants converging to zero, such that
    \begin{enumerate}[label=\roman*.]
        \item $\sup_{\eta \in \mathcal{T}_n,\, \theta \in \Theta} \Big|\EE\psi(W_i; \theta, \eta) - \EE\psi(W_i; \theta, \eta_0)\Big| = o(\tau_n)$,

        \item $\sup_{\eta \in \mathcal{T}_n,\, |\theta-\theta_0|\leq\tau_n}
        \EE \Big[\big(\psi(W_i; \theta,\eta)-\psi(W_i; \theta_0, \eta_0) \big)^2\Big]=o(1)$,

        \item $\sup_{\substack{r \in (0,1),\,\eta \in \mathcal{T}_n,|\theta-\theta_0|\leq\tau_n}}\Big| \partial^2_r\EE\psi\big(W_i;\theta_0+r(\theta-\theta_0),\eta_0+r(\eta-\eta_0)\big) \Big|=o(n^{-1/2}).$
    \end{enumerate}

    \item The second moment of the score function is finite and non-degenerate: $0<\EE(\psi^2(W_i; \theta_0, \eta_0))<\infty$.
\end{enumerate}
\end{condition}

The verification of Conditions~\ref{cond:a1} and~\ref{cond:a2} under Assumptions~\ref{ass:reg} and~\ref{ass:ml} is provided in Appendix~\ref{sec:proofNormal}, followed by the proof of Theorem~\ref{thm:normal}.

\section{Verification of conditions for asymptotic normality and proof of normality for the DML estimator}
\label{sec:proofNormal}
\textbf{Verification of Condition \ref{cond:a1}}

\begin{enumerate}[label=\textsc{Step \arabic*.},
                  leftmargin=5em,
                  labelsep=0.5em,
                  itemsep=5pt,
                  topsep=0.5\baselineskip]
\item First, we verify Condition \ref{cond:a1}(a). We have already shown that $\EE\psi(W_i; \theta_0, \eta_0)=0$ in Step 1 of Appendix~\ref{proof:orthogonal}.
In addition, by Assumption \ref{ass:reg}(a), the parameter space $\Theta$ contains an interval of length $c_1(n^{-1/2}\log n)$, where $c_1>0$ is a finite constant.

\item Next, we want to show that $\partial^2_\eta\, \EE(\psi)[\Delta\eta_1, \Delta\eta_2]$ is continuous, where
\begin{align*}
\Delta\eta_1 = \{h_{g,1}, h_{\mu,1}\} \quad \text{and} \quad \Delta\eta_2 = \{h_{g,2}, h_{\mu,2}\}
\end{align*}
are two directions of $\eta$ such that $h_{g,\cdot} = [g-g_0]_.$ and $h_{\mu,\cdot} = [\mu-\mu_0]_.$.

The first derivative with respect to $\eta$ is
\begin{align*}
\partial_\eta\,\EE(\psi)[\Delta\eta_1] 
&= \EE\left[\partial_g \psi \, h_{g,1} + \partial_u \psi \, h_{\mu,1}\right] \\
&= \EE\left[-(A_i-\mu_0)e^{\theta_0 A_i + g_0} h_{g,1} - (Y_i-e^{\theta_0 A_i + g_0}) h_{\mu,1}\right]
\end{align*}

The second derivative with respect to $\eta$ in a different direction is
\begin{align*}
\partial^2_\eta\, \EE(\psi)[\Delta\eta_2] 
&= \EE[h_{\mu,2} e^{\theta_0 A_i + g_0} h_{g,1} - h_{g,2}\underbrace{(A_i-\mu_0)e^{\theta_0 A_i + g_0}}_{=0|\bX_i} h_{g,1} + h_{g,2}e^{\theta_0 A_i + g_0} h_{\mu,1} + 0] \\
&= \EE[e^{\theta_0 A_i + g_0} \cdot (h_{\mu,2} \cdot h_{g,1} + h_{g,2} \cdot h_{\mu,1})].
\end{align*}

By Assumption \ref{ass:reg}(a) and \ref{ass:reg}(b), $\theta$ and $A_i$ are bounded; the second derivative $\partial^2_\eta\, \EE(\psi)$ is smooth since the exponential function is smooth.
Therefore, the map $(\theta, \eta) \mapsto \EE[\psi(W_i; \theta, \eta)]$ is twice continuously Gateaux differentiable.

\item We apply a first-order Taylor expansion to $\EE[\psi(W_i;\theta,\eta_0)]$ around $\theta_0$:
$$
\EE[\psi(W_i;\theta,\eta_0)]
=\EE[\psi(W_i;\theta_0,\eta_0)]
+J_0(\theta-\theta_0)+R(\theta),
$$
where
$$
J_0=\partial_\theta \EE[\psi(W_i;\theta,\eta_0)]\big|_{\theta=\theta_0}
\quad\text{and}\quad
R(\theta)=o(|\theta-\theta_0|).
$$
By Condition \ref{cond:a1}(a), $\EE[\psi(W_i;\theta_0,\eta_0)]=0$, hence
$$
\EE[\psi(W_i;\theta,\eta_0)] = J_0(\theta-\theta_0)+R(\theta).
$$
By the definition of $R(\theta)$, there exists $\delta>0$ such that for all $|\theta-\theta_0|\le\delta$,
$$
|R(\theta)|\le \tfrac12 |J_0|\,|(\theta-\theta_0)|.
$$
Therefore, for all $|\theta-\theta_0|\le\delta$, the reverse triangle inequality implies:
$$
2\left|\EE[\psi(W_i;\theta,\eta_0)]\right|
=2|J_0(\theta-\theta_0)+R(\theta)|
\ge 2|J_0|\,|\theta-\theta_0|-2|R(\theta)|
\ge |J_0|\,|\theta-\theta_0|.
$$

Next, by Condition \ref{cond:a1}(a)--(b), $\EE[\psi(W_i;\theta,\eta_0)]$ is continuous in $\theta$ on $\Theta$.
For a binary treatment, since 
$$
\EE[\psi(W_i;\theta,\eta_0)] =(e^{\theta_0}-e^\theta) \EE \left[ \frac{e^{g_0(X)}\pi_0(X)(1-\pi_0(X))} {1-\pi_0(X)+\pi_0(X)e^\theta} \right]
$$
has a strictly positive expectation factor and therefore vanishes only at $\theta=\theta_0$, we obtain that $\theta_0$ is the unique solution to $\EE[\psi(W_i;\theta,\eta_0)]=0$ on $\Theta$.

Fix the $\delta>0$ above and define
$$
c_0:=2\; \inf_{\theta\in\Theta:\,|\theta-\theta_0|\ge\delta}
\left|\EE[\psi(W_i;\theta,\eta_0)]\right|.
$$
Since $\{\theta\in\Theta:|\theta-\theta_0|\ge\delta\}$ is compact and the integrand is continuous,
the infimum is attained and is strictly positive, so $c_0>0$. Hence for all $\theta\in\Theta$ with $|\theta-\theta_0|\ge\delta$,
$$
2\left|\EE[\psi(W_i;\theta,\eta_0)]\right|\ge c_0.
$$
Combining the cases $|\theta-\theta_0|\le\delta$ and $|\theta-\theta_0|\ge\delta$ yields
$$
2\left|\EE[\psi(W_i;\theta,\eta_0)]\right|
\ge \min\{|J_0(\theta-\theta_0)|,\,c_0\}.
$$

\item To verify Condition~\ref{cond:a1}(d), we use Proposition~\ref{prop:orthogonal}, which gives $\partial_\eta\EE[\psi(W_i;\theta_0,\eta_0)][\eta-\eta_0]=0$.
\end{enumerate}

Hence, we have verified all the requirements for Condition \ref{cond:a1}.

\vspace{1em}
\noindent\textbf{Verification of Condition \ref{cond:a2}}
\begin{enumerate}[label=\textsc{Step \arabic*.},
                  leftmargin=5em,
                  labelsep=0.5em,
                  itemsep=5pt,
                  topsep=0.5\baselineskip]
\item We first verify Condition \ref{cond:a2}(a). By Assumption \ref{ass:ml}, the nuisance functions are uniformly consistent and achieve a specified $L_2$ convergence rate.
Here, $\mu$ denotes the collection $(\mu_\theta)_{\theta\in\Theta}$, where $\mu_\theta=\pi e^\theta/(1-\pi+\pi e^\theta)$ for a propensity score $\pi\in[0,1]$.
Since this transformation is uniformly Lipschitz in $\pi$ over compact $\Theta$, the convergence rates in Assumption~\ref{ass:ml} also hold for $\hat\mu_\theta$ uniformly over $\theta\in\Theta$.
Specifically, there exist $\xi_{1,n} = o(1)$ and $\xi_{2,n} = o(n^{-1/4})$ such that $\hat{\eta}$ belongs to

$$
\mathcal{T}_n=\left\{\eta=(g,\mu):
\begin{aligned}
    &\|g-g_0\|_\infty+\sup_{\theta\in\Theta}\|\mu_\theta-\mu_{0,\theta}\|_\infty\leq\xi_{1,n},\\
    &\|g-g_0\|_{P,2}+\sup_{\theta\in\Theta}\|\mu_\theta-\mu_{0,\theta}\|_{P,2}\leq\xi_{2,n}
\end{aligned}
\right\},
$$
with probability approaching 1, for each $k \in \{1, \ldots, K\}$. 
Since $\eta_0\in\mathcal{T}_n$, this verifies Condition~\ref{cond:a2}(a).

\item Then, we verify Condition \ref{cond:a2}(b), specifically that the functional class $\mathcal{F_\eta}=\{\psi(\cdot; \theta,\eta): \theta \in \Theta \}$ satisfies the uniform entropy condition.

We first show that there exists $q>2$ such that $\|F_\eta\|_{P,q}=O(1)$.
By Assumption~\ref{ass:reg}(a),(c), the binary treatment, and the restrictions defining $\mathcal{T}_n$, we have that $\exp\{\theta A_i+g\}$ is uniformly bounded over $\theta\in\Theta$ and $\eta\in\mathcal{T}_n$.
Since $|A_i-\mu_\theta|\leq1$, we have
\begin{align*}
|\psi(W_i;\theta,\eta)| &= |A_i-\mu_\theta|\,|Y_i-e^{\theta A_i+g}|\\
&\leq |Y_i|+e^{\theta A_i+g}.
\end{align*}

For the binary treatment, $\partial\mu_\theta/\partial\theta=\mu_\theta(1-\mu_\theta)$.
Therefore,
\begin{align*}
\frac{\partial}{\partial\theta}\psi(W_i;\theta,\eta) = -\mu_\theta(1-\mu_\theta)\left(Y_i-e^{\theta A_i+g}\right) -(A_i-\mu_\theta)A_i e^{\theta A_i+g}.
\end{align*}
Using $\mu_\theta(1-\mu_\theta)\leq1/4$, $|A_i-\mu_\theta|\leq1$, and $|A_i|\leq1$, we obtain
$$
\left|\frac{\partial}{\partial\theta}\psi(W_i;\theta,\eta)\right| \leq \frac14|Y_i|+\frac54e^{\theta A_i+g}.
$$
Therefore, there exists a sufficiently large fixed constant $C$ such that
$$
|\psi(W_i;\theta,\eta)| \leq C(1+|Y_i|)\qquad \text{and} \qquad \left|\frac{\partial}{\partial\theta}\psi(W_i;\theta,\eta)\right| \leq C(1+|Y_i|).
$$
Define
$$
m(W_i)=F_\eta(W_i)=C(1+|Y_i|).
$$
By Assumption~\ref{ass:reg}(b), there exists $q>2$ such that
$$
\|F_\eta\|_{P,q} \leq C(1+\|Y_i\|_{P,q})=O(1),
$$
uniformly over $\eta\in\mathcal{T}_n$.

Next, we establish the Lipschitz condition required by Example 19.7 of \cite{vandervaart00asymptotic}.
For any $\theta_1,\theta_2\in\Theta$, the mean value theorem (MVT) gives a $\theta^*$ between them such that
\begin{align*}
|\psi(W_i;\theta_1,\eta)-\psi(W_i;\theta_2,\eta)| &= \left|\frac{\partial}{\partial\theta}\psi(W_i;\theta^*,\eta)\right|\,|\theta_1-\theta_2|\\
&\leq m(W_i)|\theta_1-\theta_2|.
\end{align*}
The binary treatment defines $\mu_\theta$ smoothly throughout the interval between $\theta_1$ and $\theta_2$.

We now apply the Lipschitz bracketing bound.
Since the parameter $\theta$ is a scalar, its dimension is $d=1$.
For every finitely supported probability measure $Q$, $\|m\|_{Q,2}<\infty$.
Applying the bound with $Q$ in place of $P$ gives
$$
\mathcal{N}_{[]}\left(\varepsilon\|m\|_{Q,2},\mathcal{F}_\eta,\|\cdot\|_{Q,2}\right) \leq K\left(\frac{\operatorname{diam}(\Theta)}{\varepsilon}\right), \qquad 0<\varepsilon<\operatorname{diam}(\Theta),
$$
where $\mathcal{N}_{[]}$ denotes the bracketing number and $K$ can be chosen independently of $Q$ and $\eta$.

Since covering numbers are bounded by bracketing numbers, we have 
$$
\mathcal{N}\left(\varepsilon\|m\|_{Q,2},\mathcal{F}_\eta,\|\cdot\|_{Q,2}\right) \leq \mathcal{N}_{[]}\left(\varepsilon\|m\|_{Q,2},\mathcal{F}_\eta,\|\cdot\|_{Q,2}\right) \leq K\frac{\operatorname{diam}(\Theta)}{\varepsilon},
$$
for $0<\varepsilon<1$.
Setting $\nu=1$ and $R=K\text{diam}(\Theta)$ and taking logarithms on both sides, we obtain
\begin{equation*}
    \log \mathcal{N} \left( \varepsilon\|m\|_{Q,2}, \mathcal{F}_\eta,\|\cdot\|_{Q,2} \right) \leq \nu\log( R/\varepsilon), \qquad 0<\varepsilon<1.
\end{equation*}

Since $m=F_\eta$ and the bound above is independent of $Q$, taking the supremum over finitely supported probability measures $Q$ gives
$$
\sup_Q \log \mathcal{N} \left( \varepsilon\|F_\eta\|_{Q,2}, \mathcal{F}_\eta, \|\cdot\|_{Q,2} \right) \leq \nu\log(R/\varepsilon), \qquad 0<\varepsilon\leq1.
$$

This completes the verification of Condition~\ref{cond:a2}(b).

\item The next step is to verify Condition \ref{cond:a2}(c) (i)--(iii).
\begin{enumerate}
    \item[i.] For each $\eta = \{g, \mu\} \in \mathcal{T}_n$ and $\theta \in \Theta$, the following holds:  
    \begin{align*}
    &\EE[\psi(W_i; \theta, \eta)] - \EE[\psi(W_i; \theta, \eta_0)]  \\
    &= \EE\left[(Y_i - e^{\theta A_i + g})(A_i - \mu) - (Y_i - e^{\theta A_i + g_0})(A_i - \mu_0)\right]  \\
    &=\EE\left[(Y_i-e^{\theta A_i+g})(A_i-\mu) - (Y_i-e^{\theta A_i+g})(A_i-\mu_0)\right.  \\
    & \quad + \left.(Y_i-e^{\theta A_i+g})(A_i-\mu_0) - (Y_i-e^{\theta A_i+g_0})(A_i-\mu_0)\right]   \\
    &= \EE\left[-(Y_i-e^{\theta A_i+g})(\mu-\mu_0) + (e^{\theta A_i+g_0}-e^{\theta A_i+g})(A_i-\mu_0)\right]   \\
    & \leq \underbrace{\EE\left[\left|(Y_i-e^{\theta A_i+g})\cdot(\mu-\mu_0)\right|\right]}_{\text{Term 1}} + \underbrace{\EE\left[\left|(e^{\theta A_i+g_0}-e^{\theta A_i+g})\cdot(A_i-\mu_0)\right|\right]}_{\text{Term 2}}.
    \end{align*}

    The second equality holds because we add and subtract $(Y_i-e^{\theta A_i+g})(A_i-\mu_0)$.
    The last inequality is due to the triangle inequality.
    
    We first examine Term 1. There exists a positive constant $C_1 < \infty$ such that:

    \begin{align*}
        \EE\left[\left|(Y_i-e^{\theta A_i+g})\cdot(\mu-\mu_0)\right|\right] &\leq \norm{Y_i-e^{\theta A_i+g}}_{P,2} \cdot \norm{\mu-\mu_0}_{P,2} \\
        & \leq \left(\norm{Y_i}_{P,2} + \norm{e^{\theta A_i+g}}_{P,2}\right)\cdot \norm{\mu-\mu_0}_{P,2} \\
        &=C_1 \norm{\mu-\mu_0}_{P,2}.
    \end{align*}

    The first inequality follows from the Cauchy--Schwarz inequality.
    The second inequality is due to the triangle inequality.
    The last inequality holds due to the boundedness of $\theta$, $A_i$, $g$ and $\norm{Y_i}_{P,2}$ by Assumption \ref{ass:reg}(a)--(c).

    Next, we apply the MVT to the exponential term in Term 2:
    \begin{equation} \label{eq:mvt_g}
        |e^{\theta A_i+g}-e^{\theta A_i+g_0}| = e^{\theta A_i+ g^*}\cdot|g-g_0|,
    \end{equation}
    where $g^*$ lies between $g$ and $g_0$. 

    Substitute equation~(\ref{eq:mvt_g}) into Term 2 and then apply the Cauchy--Schwarz inequality:

    \begin{align*}
        \EE\left|(e^{\theta A_i+g_0}-e^{\theta A_i+g})\cdot(A_i-\mu_0)\right|
        &= \EE\left| (A_i-\mu_0)e^{\theta A_i+g^*}\cdot(g-g_0) \right| \\
        &\leq \norm{(A_i-\mu_0)e^{\theta A_i+g^*}}_{P,2} \cdot \norm{g-g_0}_{P,2}.
    \end{align*}
    
    Since $A_i$, $g$, and $\mu_0$ are bounded by Assumption \ref{ass:reg}(b), \ref{ass:ml}(a) and Remark \ref{remark:mu}, the exponential term is bounded, and there exists a positive constant $C_2<\infty$ such that:
    \begin{align*}
        \norm{(A_i-\mu_0)e^{\theta A_i+g^*}}_{P,2} \leq C_2.
    \end{align*}
    Therefore, 
    $$
        \EE\left|(e^{\theta A_i+g_0}-e^{\theta A_i+g})\cdot(A_i-\mu_0)\right| \leq C_2\norm{g-g_0}_{P,2}.
    $$
    
    Combining the bounds in Terms 1 and 2 and letting $C=\max(C_1, C_2)$, we obtain:
    \begin{align*}
        \sup_{\eta \in \mathcal{T}_n,\, \theta \in \Theta} \left|\EE(\psi(W_i; \theta, \eta)) - \EE(\psi(W_i; \theta, \eta_0))\right| \leq C(\norm{g-g_0}_{P,2}+\norm{\mu-\mu_0}_{P,2}).
    \end{align*}
    
    By Assumption \ref{ass:ml}(b) and Remark \ref{remark:mu}, there exists a sequence $\xi_{2,n}=o(n^{-1/4})$ that bounds the sum of the $L_2$ nuisance estimation errors uniformly over $\mathcal{T}_n$.
    Choose a sequence $\tau_n$ such that
    \begin{align*}
        \xi_{2,n}=o(\tau_n)
        \qquad\text{and} \qquad
        \tau_n=o(n^{-1/4}).    
    \end{align*}
    We additionally require $\tau_n\geq n^{-1/2}\log n$, which is compatible with the preceding restrictions because both $\xi_{2,n}$ and $n^{-1/2}\log n$ are $o(n^{-1/4})$.
    It then follows that
    $$
        \sup_{\eta \in \mathcal{T}_n,\; \theta \in \Theta} \left|\EE(\psi(W_i; \theta, \eta)) - \EE(\psi(W_i; \theta, \eta_0))\right| =o(\tau_n).
    $$

    \item[ii.] By adding and subtracting $\EE[\psi(W_i; \theta_0,\eta)]$, we obtain the following inequality:
    
    \begin{align} \label{eq:score.diff.2}
        &\EE\left[\left(\psi(W_i; \theta, \eta) - \psi(W_i; \theta_0, \eta_0)\right)^2\right] \nonumber \\
        &= \EE\left[\left( \left(\psi(W_i; \theta, \eta) - \psi(W_i; \theta, \eta_0) \right) + \left(\psi(W_i; \theta, \eta_0) - \psi(W_i; \theta_0, \eta_0)\right) \right)^2\right] \nonumber \\
        &\leq 2\underbrace{\EE\left[ \left(\psi(W_i; \theta, \eta) - \psi(W_i; \theta, \eta_0) \right)^2 \right]}_{\text{Term 1: error due to }\eta}
        + 2\underbrace{\EE\left[ \left(\psi(W_i; \theta, \eta_0) - \psi(W_i; \theta_0, \eta_0) \right)^2 \right]}_{\text{Term 2: error due to }\theta}.
    \end{align}
    The last inequality holds as $(a+b)^2 \leq 2a^2+2b^2$. 

    First, we examine Term 1 in equation~(\ref{eq:score.diff.2}) and obtain the following by adding and subtracting $(Y_i-e^{\theta A_i+g})(A_i-\mu_0)$:
    \begin{align*}
        &\EE\left[(\psi(W_i; \theta, \eta) - \psi(W_i; \theta, \eta_0))^2\right] \\
        &= \EE\left[ \left(  (Y_i-e^{\theta A_i+g})(A_i-\mu) - (Y_i -e^{\theta A_i+g_0})(A_i-\mu_0) \right)^2\right]  \\
        &= \EE \left[\left(  (Y_i-e^{\theta A_i+g})(A_i-\mu) - (Y_i-e^{\theta A_i+g})(A_i-\mu_0)\right.\right. \\
        & \qquad+ \left.\left.(Y_i-e^{\theta A_i+g})(A_i-\mu_0) - (Y_i -e^{\theta A_i+g_0})(A_i-\mu_0) \right)^2 \right]  \\
        &= \EE \left[ \left( -(Y_i-e^{\theta A_i+g})(\mu-\mu_0)
        + (A_i-\mu_0)\left(e^{\theta A_i+g_0}-e^{\theta A_i+g}\right)\right)^2\right] \\
        & \leq \underbrace{2\,\EE \left[ \left( (Y_i-e^{\theta A_i+g})(\mu-\mu_0) \right)^2 \right]}_{\text{Term 1.a: error from }\mu}
        + \underbrace{2\,\EE \left[ \left( (A_i-\mu_0)\left(e^{\theta A_i+g_0}-e^{\theta A_i+g}\right) \right)^2\right]}_{\text{Term 1.b: error from }g}
    \end{align*}
    Again, the last inequality holds because $(a+b)^2 \leq 2a^2+2b^2$.

    For Term 1.a, there exists a sequence $\epsilon_{n,\mu} \to 0$ such that $|\mu - \mu_0| \leq \epsilon_{n,\mu}$ by Remark \ref{remark:mu}.
    We pull this error bound out of the expectation:
    \begin{align*}
        2\,\EE \left[ (Y_i-e^{\theta A_i+g})^2(\mu-\mu_0)^2 \right]
        &\leq 2\epsilon_{n,\mu}^2 \cdot \EE \left[ (Y_i-e^{\theta A_i+g})^2 \right] \leq 2 \epsilon_{n,\mu}^2 C_\mu = o(1).
    \end{align*}
    The last inequality holds because $A_i, \mu_0$, and $g$ are uniformly bounded by Assumption \ref{ass:reg}(a)--(c).
    Thus, the second moment of the residual is bounded by a constant $C_\mu$.
    
    Similarly, for Term 1.b, there exists a sequence $\epsilon_{n,g} \to 0$ such that $|g - g_0| \leq \epsilon_{n,g}$ by Assumption \ref{ass:ml}(a).
    First, we apply the MVT to the exponential term:
    $$
        |e^{\theta A_i+g}-e^{\theta A_i+g_0}| \leq e^{\theta A_i+g^*}|g-g_0|.
    $$
    Substituting this into the expectation and pulling out the error bound $\epsilon_{n,g}^2$ gives:
    \begin{align*}
        &2\,\EE \left[ (A_i-\mu_0)^2 e^{2(\theta A_i+g^*)} (g-g_0)^2 \right] \\
        &\quad\leq 2\epsilon_{n,g}^2 \cdot \EE \left[ (A_i-\mu_0)^2 e^{2(\theta A_i+g^*)} \right] \\
        &\quad\leq 2 \epsilon_{n,g}^2 C_g = o(1).
    \end{align*}
    The last inequality holds because, again, $A_i, \mu_0$, and $g$ are uniformly bounded by Assumption \ref{ass:reg}(a)--(c) such that the expectation term is bounded by a constant $C_g$.

    For Term 2, the MVT gives
    a $\theta^*$ between $\theta$ and $\theta_0$ such that
    \begin{align*}
        \left|\psi(W_i;\theta,\eta_0)-\psi(W_i;\theta_0,\eta_0) \right|
        =
        \left|\partial_\theta\psi(W_i;\theta^*,\eta_0)\right|\cdot|\theta-\theta_0|.
    \end{align*}
    The derivative is
    \begin{align*}
        \partial_\theta\psi(W_i;\theta^*,\eta_0)
        &= -\mu_{\theta^*}(1-\mu_{\theta^*})
        \left(Y_i-\exp\{\theta^*A_i+g_0\}\right) \\
        &\quad-(A_i-\mu_{\theta^*})A_i \exp\{\theta^*A_i+g_0\}.
    \end{align*}
    Since $0\leq\mu_\theta^*\leq1$,
    $\mu_\theta^*(1-\mu_\theta^*)\leq1/4$, and $A_i$,
    $\theta^*$, and $g_0$ are bounded, there exists a constant
    $C_\lambda<\infty$ such that
    \begin{align*}
        \sup_{\theta\in\Theta} \left| \partial_\theta \psi(W_i;\theta^*,\eta_0) \right| \leq
        \frac{1}{4}|Y_i| + \frac{5}{4}C_\lambda.
    \end{align*}
    Therefore, Assumption~\ref{ass:reg}(a)--(c) implies that
    \begin{align*}
        \EE\left[ \sup_{\theta^*\in\Theta} \left| \partial_\theta\psi(W_i;\theta^*,\eta_0) \right|^2 \right]
        \leq C_\theta^2<\infty.
    \end{align*}
    It follows that Term 2 in equation~(\ref{eq:score.diff.2}) satisfies
    \begin{align*}
        \EE\left[ \left( \psi(W_i;\theta,\eta_0) - \psi(W_i;\theta_0,\eta_0) \right)^2 \right] & = |\theta-\theta_0|^2
        \EE\left[ \left( \partial_\theta\psi(W_i;\theta^*,\eta_0) \right)^2 \right] \\
        &\leq |\theta-\theta_0|^2 \EE\left[ \sup_{\theta^*\in\Theta} \left| \partial_\theta\psi(W_i;\theta^*,\eta_0) \right|^2 \right] \\
        &\leq \tau_n^2C_\theta^2 = o(n^{-1/2}).
    \end{align*}

    Therefore, combining Terms 1.a, 1.b, and 2, we have:
    \begin{align*}
        \sup_{\eta\in \mathcal{T}_n,\!|\theta-\theta_0|\leq\tau_n}\EE\left[\left(\psi(W_i; \theta, \eta) - \psi(W_i; \theta_0, \eta_0)\right)^2\right]
        &\leq C_\mu C_gC_\theta^2(\epsilon_{n,\mu}^2 + \epsilon_{n,g}^2 + \tau_n^2) = o(1).
    \end{align*}

    \item[iii.] For $r\in(0,1)$, define $\theta_r=\theta_0+r(\theta-\theta_0)$ and $g_r=g_0+r(g-g_0)$. Since $\mu$ depends on $\theta$, we define
    $$
    \mu_{r,\theta} = \mu_{0,\theta}+r(\mu_\theta-\mu_{0,\theta}), \qquad \theta\in\Theta,
    $$
    and let
    $$
    \eta_r = (g_r,\mu_{r,\cdot}) = \eta_0+r(\eta-\eta_0).
    $$
    Thus, the nuisance value in the score is
    $$
    \mu_r = \mu_{r,\theta_r} = (1-r)\mu_{0,\theta_r}+r\mu_{\theta_r},
    $$
    and we have 
    $$
    \psi(W;\theta_r,\eta_r) = (A-\mu_r)\left(Y-e^{\theta_r A+g_r}\right).
    $$

    For every $\eta=(g,\mu)\in\mathcal{T}_n$, we obtain the following bound from STEP 1:
    $$
    \|g-g_0\|_\infty+\sup_{\theta\in\Theta}\|\mu_\theta-\mu_{0,\theta}\|_\infty \leq \xi_{1,n}=o(1)
    $$
    and
    $$
    \|g-g_0\|_{P,2}+\sup_{\theta\in\Theta}\|\mu_\theta-\mu_{0,\theta}\|_{P,2} \leq \xi_{2,n}=o(n^{-1/4}).
    $$
    We use the choice $\tau_n=o(n^{-1/4})$ from part (i) and consider $|\theta-\theta_0|\leq\tau_n$.

    We first obtain bounds for the derivatives of $\mu_r$.
    Let $\mu_r'$ and $\mu_r''$ denote the first and second derivatives with respect to $r$.
    All partial derivatives are evaluated at $\theta=\theta_r$.
    By the product rule and chain rule,
    \begin{align*}
    \mu_r' &= \mu_{\theta_r}-\mu_{0,\theta_r} + (\theta-\theta_0)\left((1-r)\frac{\partial\mu_{0,\theta}}{\partial\theta}+r\frac{\partial\mu_\theta}{\partial\theta}\right),
    \end{align*}
    and
    \begin{align*}
    \mu_r'' &= 2(\theta-\theta_0)\left(\frac{\partial\mu_\theta}{\partial\theta}-\frac{\partial\mu_{0,\theta}}{\partial\theta}\right) + (\theta-\theta_0)^2\left((1-r)\frac{\partial^2\mu_{0,\theta}}{\partial\theta^2}+r\frac{\partial^2\mu_\theta}{\partial\theta^2}\right).
    \end{align*}

    For the binary treatment, we can obtain that
    $$
    \frac{\partial\mu_\theta}{\partial\theta} = \mu_\theta(1-\mu_\theta), \qquad \frac{\partial^2\mu_\theta}{\partial\theta^2} = \mu_\theta(1-\mu_\theta)(1-2\mu_\theta),
    $$
    and the same identities hold for $\mu_{0,\theta}$.
    Since $\mu_\theta$ and $\mu_{0,\theta}$ take values in $[0,1]$, their first and second derivatives with respect to $\theta$ are uniformly bounded.
    Moreover, for every $\theta$,
    \begin{align*}
    \left|\frac{\partial\mu_\theta}{\partial\theta}-\frac{\partial\mu_{0,\theta}}{\partial\theta}\right| &= \left|\mu_\theta(1-\mu_\theta)-\mu_{0,\theta}(1-\mu_{0,\theta})\right|\\
    &= |\mu_\theta-\mu_{0,\theta}|\,|1-\mu_\theta-\mu_{0,\theta}|\\
    &\leq |\mu_\theta-\mu_{0,\theta}|.
    \end{align*}
    Consequently, we obtain the following: 
    \begin{equation}
    \label{eq:mu-path-first-bound}
    |\mu_r'| \leq |\mu_{\theta_r}-\mu_{0,\theta_r}|+C|\theta-\theta_0|,
    \end{equation}
    and
    \begin{equation}
    \label{eq:mu-path-second-bound}
    |\mu_r''| \leq 2|\theta-\theta_0|\,|\mu_{\theta_r}-\mu_{0,\theta_r}|+C|\theta-\theta_0|^2.
    \end{equation}
    Here, $C$ denotes a sufficiently large finite constant independent of $r$, $\theta$, $\eta$, and $n$.

    Now we want to examine 
    \begin{align*}
        \frac{\partial}{\partial r}\EE [\psi(W;\theta_r,\eta_r)], \qquad \text{and} \qquad \frac{\partial^2}{\partial r^2}\EE [\psi(W;\theta_r,\eta_r)]. 
    \end{align*}
    The first derivative is
    \begin{align*}
    \frac{\partial}{\partial r}\EE [\psi(W;\theta_r,\eta_r)] &= -\EE \left[(A-\mu_r)e^{\theta_r A+g_r}\left(A(\theta-\theta_0)+g-g_0\right)\right]\\
    &\quad - \EE \left[\mu_r'\left(Y-e^{\theta_r A+g_r}\right)\right].
    \end{align*}
    Differentiating again gives
    \begin{align*}
    \frac{\partial^2}{\partial r^2}\EE [\psi(W;\theta_r,\eta_r)] 
    &= \underbrace{2\EE \left[\mu_r'e^{\theta_r A+g_r}\left(A(\theta-\theta_0)+g-g_0\right)\right]}_{\text{Term 1}}\\
    &\quad - \underbrace{\EE \left[(A-\mu_r)e^{\theta_r A+g_r}\left(A(\theta-\theta_0)+g-g_0\right)^2\right]}_{\text{Term 2}}\\
    &\quad - \underbrace{\EE \left[\mu_r''\left(Y-e^{\theta_r A+g_r}\right)\right]}_{\text{Term 3}}.
    \end{align*}
    By the triangle inequality,
    $$
    \left|\frac{\partial^2}{\partial r^2}\EE [\psi(W;\theta_r,\eta_r)]\right| \leq |\text{Term 1}|+|\text{Term 2}|+|\text{Term 3}|.
    $$

    First, we examine Term 1.
    By Assumption~\ref{ass:reg}(a-c), we have that $|A|$, $\exp\{\theta_r A+g_r\}$, $|A-\mu_r|$, and $|A(\theta-\theta_0)+g-g_0|$ are uniformly bounded.
    Combining this bound with equation~\eqref{eq:mu-path-first-bound}, we obtain
    \begin{align*}
    |\text{Term 1}| &\leq 2\EE \left[|\mu_r'|e^{\theta_r A+g_r}|A(\theta-\theta_0)+g-g_0|\right]\\
    &\leq C\, \EE \left[\left(|\mu_{\theta_r}-\mu_{0,\theta_r}|+|\theta-\theta_0|\right)\left(|\theta-\theta_0|+|g-g_0|\right)\right]\\
    &= C\Big(|\theta-\theta_0|^2+|\theta-\theta_0|\cdot\EE [|g-g_0|]\\
    & \quad +|\theta-\theta_0|\cdot\EE [|\mu_{\theta_r}-\mu_{0,\theta_r}|]   +\EE [|\mu_{\theta_r}-\mu_{0,\theta_r}|\,\cdot|g-g_0|]\Big)\\
    &\leq C\Big(|\theta-\theta_0|^2+|\theta-\theta_0|\cdot\|g-g_0\|_{P,2}\\
    & \quad +|\theta-\theta_0|\cdot\|\mu_{\theta_r}-\mu_{0,\theta_r}\|_{P,2}  +\|\mu_{\theta_r}-\mu_{0,\theta_r}\|_{P,2}\cdot\|g-g_0\|_{P,2}\Big)\\
    &\leq C\left(\tau_n^2+\tau_n\xi_{2,n}+\xi_{2,n}^2\right) = o(n^{-1/2}).
    \end{align*}
    The second-to-last inequality follows from the Cauchy--Schwarz inequality.

    Similarly, for Term 2,
    \begin{align*}
    |\text{Term 2}| &\leq \EE \left[|A-\mu_r|e^{\theta_r A+g_r}\left(A(\theta-\theta_0)+g-g_0\right)^2\right]\\
    &\leq C \, \EE \left[\left(A(\theta-\theta_0)+g-g_0\right)^2\right]\\
    &\leq C\left(|\theta-\theta_0|^2+2|\theta-\theta_0|\EE [|g-g_0|]+\EE [(g-g_0)^2]\right)\\
    &\leq C\left(|\theta-\theta_0|^2+2|\theta-\theta_0|\|g-g_0\|_{P,2}+\|g-g_0\|_{P,2}^2\right)\\
    &\leq C\left(\tau_n^2+\tau_n\xi_{2,n}+\xi_{2,n}^2\right)=o(n^{-1/2}).
    \end{align*}

    Finally, we examine Term 3.
    Since
    $$
    \EE [Y\mid A,X]=e^{\theta_0 A+g_0},
    $$
    and $\mu_r''$ is a function of $X$, iterated expectation gives
    \begin{align*}
    \EE \left[\mu_r''\left(Y-e^{\theta_r A+g_r}\right)\right] &= \EE \left[\mu_r''\EE \left[Y-e^{\theta_r A+g_r}\mid A,X\right]\right]\\
    &= \EE \left[\mu_r''\left(e^{\theta_0 A+g_0}-e^{\theta_r A+g_r}\right)\right].
    \end{align*}
    Both exponential terms are uniformly bounded.
    Hence, using equation~\eqref{eq:mu-path-second-bound},
    \begin{align*}
    |\text{Term 3}| &\leq \EE \left[|\mu_r''| \cdot \left|e^{\theta_0 A+g_0}-e^{\theta_r A+g_r}\right|\right]\\
    &\leq C \, \EE [|\mu_r''|]\\
    &\leq C\left(|\theta-\theta_0|^2+|\theta-\theta_0| \cdot \EE [|\mu_{\theta_r}-\mu_{0,\theta_r}|]\right)\\
    &\leq C\left(|\theta-\theta_0|^2+|\theta-\theta_0| \cdot \|\mu_{\theta_r}-\mu_{0,\theta_r}\|_{P,2}\right)\\
    &\leq C\left(\tau_n^2+\tau_n\xi_{2,n}\right)=o(n^{-1/2}).
    \end{align*}

    Combining Terms 1, 2, and 3 yields
    \begin{align*}
    &\sup_{\substack{r\in(0,1),\,\eta\in\mathcal{T}_n,\!|\theta-\theta_0|\leq\tau_n}}\left|\frac{\partial^2}{\partial r^2} \EE [\psi(W;\theta_r,\eta_r)]\right| \leq C \left( \tau_n^2+\tau_n\xi_{2,n}+\xi_{2,n}^2 \right)=o(n^{-1/2}),
    \end{align*}
    which verifies Condition~\ref{cond:a2}(c)iii.
    
\end{enumerate}

\item Finally, we verify Condition \ref{cond:a2}(d). 
This is directly given by Assumption \ref{ass:reg}(d): $c_0\leq \EE(\psi^2(W_i; \theta_0, \eta_0))\leq c_1$, where $0<c_0\leq c_1<\infty$.
\end{enumerate}

We use this result to prove Theorem~\ref{thm:normal}.

\begin{proof}[Proof of Theorem~\ref{thm:normal}]
The preceding verification of Conditions~\ref{cond:a1}--\ref{cond:a2} establishes that Assumptions 3.3 and 3.4 of \cite{chernozhukov18double} hold for the score function in equation~(\ref{eq:score}).
Therefore, Theorem 3.3 of \cite{chernozhukov18double} applies to the proposed DML estimator and yields the conclusions stated in Theorem~\ref{thm:normal}.
\end{proof}

\begin{proof}[Proof of Corollary~\ref{cor:variance}]
By the consistency of \(\hat\theta\), uniform consistency of the nuisance estimators, and smoothness of the score and its derivative, we have
$$
    \frac{1}{n}\sum_{i=1}^n \left( \psi (W_i;\hat\theta, \hat\eta)^2-\psi(W_i; \theta_0, \eta_0)^2 \right) = o_p(1)
$$
and
$$
    \hat J - \frac{1}{n}\sum_{i=1}^n \partial_\theta\psi(W_i;\theta_0,\eta_0) = o_p(1). 
$$
The law of large numbers then gives
$$
    n^{-1}\sum_{i=1}^n\psi (W_i;\hat\theta, \hat\eta)^2 \xrightarrow{p} \EE[\psi (W_i; \theta_0, \eta_0)^2]
$$
and $\hat J\xrightarrow{p}J_0$. 
Since $J_0$ is bounded away from zero by Assumption \(\ref{ass:reg}(d)\), the continuous mapping theorem yields $\hat\sigma^2\xrightarrow{p}\sigma_0^2$.
\end{proof}

\section{Proof of double robustness of the estimating equation} \label{sec:proofDR}
We formally show that the score function in equation~(\ref{eq:score}) is doubly robust at the true target parameter $\theta_0$.
Let $g$ and $\mu$ denote arbitrary and possibly misspecified limits of the two nuisance estimators and $\eta=(g, \mu)$.
We show that
\begin{align*}
    \EE\big[\psi(W_i;\theta_0,\eta)\big]=0
\end{align*}
if either $g=g_0$ or $\mu=\mu_0$.  
Consequently, under the identification condition in Condition~\ref{cond:a1}(c) and the remaining regularity conditions, the estimating equation remains unbiased for $\theta_0$ when either nuisance component is correctly specified.

\begin{proof}
First suppose that the outcome regression is correctly specified, so that $g=g_0$, while $\mu$ may be misspecified.
By the partially linear Poisson model in equation~(\ref{eq:plpois}),
\begin{align*}
    \EE\left[Y_i-e^{\theta_0A_i+g_0(\bX_i)}\mid A_i,\bX_i\right]=0.
\end{align*}
The law of iterated expectations therefore gives
\begin{align*}
    \EE\big[\psi(W_i;\theta_0,(g_0,\mu))\big]
    &=\EE\left[
        \big(A_i-\mu(\bX_i)\big)
        \EE\left[Y_i-e^{\theta_0A_i+g_0(\bX_i)}\mid A_i,\bX_i\right]
      \right] =0.
\end{align*}

Next suppose that the treatment nuisance is correctly specified, so that $\mu=\mu_0$, while $g$ may be misspecified.
Conditional on $(A_i,\bX_i)$,
\begin{align*}
    \EE\left[Y_i-e^{\theta_0A_i+g(\bX_i)}\mid A_i,\bX_i\right]
    &=e^{\theta_0A_i+g_0(\bX_i)}-e^{\theta_0A_i+g(\bX_i)} \\
    &=e^{\theta_0A_i}
      \left(e^{g_0(\bX_i)}-e^{g(\bX_i)}\right).
\end{align*}
Hence,
\begin{align*}
    \EE\big[\psi(W_i;\theta_0,(g,\mu_0))\big]
    &=\EE\left[
       \left(e^{g_0(\bX_i)}-e^{g(\bX_i)}\right)
       \EE\left[(A_i-\mu_0(\bX_i))e^{\theta_0A_i}\mid\bX_i\right]
      \right].
\end{align*}
By the definition of $\mu_0$ in Section~\ref{sec:score}, 
\begin{align*}
    \mu_0(\bX_i)
    =\frac{\EE[A_ie^{\theta_0A_i}\mid\bX_i]}
           {\EE[e^{\theta_0A_i}\mid\bX_i]},
\end{align*}
so the inner conditional expectation is zero:
\begin{align*}
    \EE\left[(A_i-\mu_0(\bX_i))e^{\theta_0A_i}\mid\bX_i\right]
     =\EE[A_ie^{\theta_0A_i}\mid\bX_i] -\mu_0(\bX_i)\EE[e^{\theta_0A_i}\mid\bX_i] =0.
\end{align*}
Thus, $\EE[\psi(W_i;\theta_0,(g,\mu_0))]=0$ for arbitrary $g$.

When $A_i$ is binary and $\pi_0(\bX_i)=\PP(A_i=1\mid\bX_i)$, we have
\begin{align*}
    \mu_0(\bX_i) =\frac{\pi_0(\bX_i)e^{\theta_0}} {1-\pi_0(\bX_i)+\pi_0(\bX_i)e^{\theta_0}}.
\end{align*}
Thus, this cancellation is exactly
\begin{align*}
    \EE\left[(A_i-\mu_0(\bX_i))e^{\theta_0A_i}\mid\bX_i\right]&=\pi_0(\bX_i)e^{\theta_0}
      -\mu_0(\bX_i)\left(1-\pi_0(\bX_i)
      +\pi_0(\bX_i)e^{\theta_0}\right)=0,
\end{align*}
Therefore, the population estimating equation is unbiased at $\theta_0$ when either nuisance component is correctly specified, which establishes the claimed double robustness.
\end{proof}

\section{Simulation study: implementation and hyperparameter selection}
\label{sec:sim_appendix}

This appendix provides the implementation of nuisance-function estimators, definitions of comparative methods, and hyperparameter selection procedures used in the simulation study.
The data-generating process is specified in Section~\ref{sec:sim_dgp}, and the numerical results are reported in Tables~\ref{tab:sim_g1} and~\ref{tab:sim_g2}.
The general DML estimation and inference procedure is described in Section~\ref{sec:procedure}.

\subsection{Estimator implementation}
We evaluated the proposed DML estimator using either logistic regression or a random forest to estimate the propensity score.
For both models, the outcome nuisance function was estimated using the partially linear Poisson neural network architecture introduced in Section~\ref{subsec:nuest}.
The covariate branch contained two equal-width hidden layers, with GELU activation and dropout applied after each hidden layer, while treatment entered through a separate linear node.
Parameters in the covariate branch and the treatment coefficient were jointly optimized using Adam and the Poisson negative log-likelihood described in Section~\ref{subsec:nuest}.

The logistic regression propensity model included the four covariates as linear terms on the logit scale.
The random forest propensity model used the same four covariates as predictors.
We applied the cross-fitting procedure in Section~\ref{sec:procedure} separately with $K=2$ and $K=5$.
Within each evaluation dataset, the hyperparameters selected from the reserved tuning datasets remained the same across the cross-fitting folds.
For numerical stability, the estimated propensity scores were truncated to $[10^{-6},1-10^{-6}]$.

For each propensity score model, we also evaluated a DML estimator in which the known $g_s(\bX_i)$ from the data-generating process replaced the nuisance function estimated by neural networks.
The propensity score remained estimated from the observed treatment and covariates.
This estimator serves as a benchmark since it isolates the contribution of nuisance estimation error from $g_0(\bX_i)$.

The linearly adjusted Poisson GLM included treatment and the four covariates as additive linear predictors.
Its fitted treatment coefficient was reported as the estimate of $\theta_0$.

The naive partially linear Poisson neural network estimator used the same architecture and selected hyperparameters as the DML outcome model but was fitted once to the complete evaluation dataset.
The coefficient on its linear treatment node was reported directly as the estimate of $\theta_0$, without cross-fitting or the Neyman-orthogonal estimating equation.

For AIPW, a full-sample Poisson GLM containing treatment and the four covariates as additive linear predictors was used to estimate the conditional mean count under each treatment level, denoted by $\hat m_a(\bX_i)$ for $a\in\{0,1\}$.
The propensity score was estimated using five-fold cross-fitting with either logistic regression or a random forest.
Let $\hat p_1(\bX_i)=\hat\pi_i$ and $\hat p_0(\bX_i)=1-\hat\pi_i$.
The marginal counterfactual means and the corresponding log count mean ratio were estimated as
\begin{align*}
\hat\nu_a &= \frac{1}{n}\sum_{i=1}^n \left[\hat m_a(\bX_i) +
\frac{\mathbf{1}\{A_i=a\}}{\hat p_a(\bX_i)} \left\{Y_i-\hat m_a(\bX_i)\right\}\right], \qquad a\in\{0,1\}.
\end{align*}

Under the partially linear Poisson model, $\nu_1/\nu_0=\exp\{\theta_0\}$, so the AIPW estimator 
\begin{align*}
\hat\theta_{\mathrm{AIPW}} =
\log\left(\frac{\hat\nu_1}{\hat\nu_0}\right)
\end{align*}
targets the same $\theta_0$ as the other methods.

\subsection{Hyperparameter selection}
For each simulation setting, we reserved the first five of the 500 simulated datasets for hyperparameter tuning and excluded these datasets from the Monte Carlo performance summaries.
Each candidate configuration was evaluated using five-fold cross-validation within each tuning dataset.
The validation losses were averaged across the five folds and five tuning datasets, and the configuration with the lowest average loss was fixed for analyzing the remaining 495 datasets.
This cross-validation for hyperparameter selection was separate from the two-fold or five-fold DML cross-fitting subsequently performed within each evaluation dataset.

For the neural network, the candidate grid included hidden-layer widths of 8 or 16, dropout probabilities of 0.01 or 0.05, and 500 or 1,000 training epochs.
The two hidden layers always had the same width, and the learning rate was fixed at 0.01.
Candidates were compared using the held-out Poisson negative log-likelihood for the complete partially linear predictor $\theta A_i+g(\bX_i)$.
The selected configuration was used for the cross-fitted outcome nuisance models in DML and for the corresponding full-sample naive neural network estimator.

For the random forest, the candidate grid included 50, 100, or 200 trees; maximum tree depths of 4, 8, or 16; minimum split sizes of 2, 5, or 10; and minimum leaf sizes of 2, 5, or 10.
At each split, the number of candidate predictors was determined by the square root of the total number of covariates, and each tree was fitted using a bootstrap sample.
The candidates were compared using held-out treatment log loss averaged across the validation folds and the five tuning datasets.

The random forest was tuned separately for each simulation setting, and the selected configuration was fixed for analyzing the remaining 495 datasets.

\section{Empirical study: data description, implementation, and diagnostics}
\label{sec:emp_appendix}

This appendix provides the baseline descriptive statistics, hyperparameter selection, inference details, and inverse probability weight diagnostics for the CCASAnet analysis.

The data were used with provider consent and are available upon request to the authors and CCASAnet, subject to provider approval.

\subsection{Baseline characteristics}
\label{sec:emp_baseline}
Table~\ref{tab:emp_baseline} summarizes the baseline characteristics of the 7,141 ART-naive patients in the analytic sample, stratified by the ART regimen that patients started.
The CCASAnet cohort underlying the synthetic data includes patients from six countries; however, the synthetic data only consist of five countries.
To improve estimation stability, categories with small sample sizes were combined in the models.
For country, Argentina and Mexico were combined into a single category.
For HIV acquisition mode, the Perinatal, Transfusion nonhemophilia related, and Other categories were combined into a single Other category.

\begin{table}[htbp]
\caption{Baseline characteristics of the CCASAnet analytic sample by initial ART regimen.}
\label{tab:emp_baseline}
\begin{center}
\small
\begin{tabular}{lccc}
\toprule
Characteristic & Non-DTG & DTG & Overall \\
 & ($N=4,918$) & ($N=2,223$) & ($N=7,141$) \\
\midrule
\multicolumn{4}{l}{\textbf{Time of ART initiation (years since January 1, 2013)}} \\
\quad Mean (SD) & 3.93 (2.35) & 7.56 (2.19) & 5.06 (2.85) \\
\quad Median [Min, Max] & 3.81 [0.00, 12.08] & 7.72 [0.03, 12.66] & 5.04 [0.00, 12.66] \\
\addlinespace[3pt]
\multicolumn{4}{l}{\textbf{Baseline CD4 count (cells/mm$^3$)}} \\
\quad Mean (SD) & 373 (221) & 356 (200) & 367 (215) \\
\quad Median [Min, Max] & 348 [16.0, 1,220] & 342 [19.0, 1,080] & 346 [16.0, 1,220] \\
\addlinespace[3pt]
\multicolumn{4}{l}{\textbf{Age at baseline (years)}} \\
\quad Mean (SD) & 35.0 (11.0) & 33.7 (11.0) & 34.6 (11.0) \\
\quad Median [Min, Max] & 31.2 [20.0, 74.5] & 30.0 [20.0, 71.1] & 31.0 [20.0, 74.5] \\
\addlinespace[3pt]
\multicolumn{4}{l}{\textbf{School level}} \\
\quad Level 0 & 264 (5.4\%) & 158 (7.1\%) & 422 (5.9\%) \\
\quad Level 1 & 314 (6.4\%) & 140 (6.3\%) & 454 (6.4\%) \\
\quad Level 2 & 643 (13.1\%) & 415 (18.7\%) & 1,058 (14.8\%) \\
\quad Level 3 & 1,419 (28.9\%) & 901 (40.5\%) & 2,320 (32.5\%) \\
\quad Level 4 & 1,902 (38.7\%) & 515 (23.2\%) & 2,417 (33.8\%) \\
\quad Level 5 & 376 (7.6\%) & 94 (4.2\%) & 470 (6.6\%) \\
\addlinespace[3pt]
\multicolumn{4}{l}{\textbf{Sex}} \\
\quad Female & 936 (19.0\%) & 460 (20.7\%) & 1,396 (19.5\%) \\
\quad Male & 3,982 (81.0\%) & 1,763 (79.3\%) & 5,745 (80.5\%) \\
\addlinespace[3pt]
\multicolumn{4}{l}{\textbf{Country}} \\
\quad Argentina & 91 (1.9\%) & 4 (0.2\%) & 95 (1.3\%) \\
\quad Brazil & 617 (12.5\%) & 923 (41.5\%) & 1,540 (21.6\%) \\
\quad Chile & 3,017 (61.3\%) & 837 (37.7\%) & 3,854 (54.0\%) \\
\quad Haiti & 693 (14.1\%) & 422 (19.0\%) & 1,115 (15.6\%) \\
\quad Mexico & 500 (10.2\%) & 37 (1.7\%) & 537 (7.5\%) \\
\addlinespace[3pt]
\multicolumn{4}{l}{\textbf{HIV acquisition mode}} \\
\quad Bisexual & 576 (11.7\%) & 346 (15.6\%) & 922 (12.9\%) \\
\quad Heterosexual contact & 1,585 (32.2\%) & 853 (38.4\%) & 2,438 (34.1\%) \\
\quad Homosexual contact & 2,590 (52.7\%) & 978 (44.0\%) & 3,568 (50.0\%) \\

\quad Perinatal & 26 (0.5\%) & 10 (0.4\%) & 36 (0.5\%) \\
\quad Transfusion nonhemophilia related & 13 (0.3\%) & 1 (0.0\%) & 14 (0.2\%) \\
\quad Other & 128 (2.6\%) & 35 (1.6\%) & 163 (2.3\%) \\
\bottomrule
\end{tabular}
\end{center}
\parbox{\textwidth}{\scriptsize Values are mean (standard deviation) and median [minimum, maximum] for continuous variables, or $n$ (\%) for categorical variables. Calendar time is measured in years since January 1, 2013. School level is reported as categorical levels 0--5. Country categories are displayed individually; Argentina and Mexico were combined in the models. HIV acquisition categories are displayed individually; the models combined the categories other than bisexual, heterosexual contact, and homosexual contact into the Other category.}
\end{table}

\subsection{Cross-fitting and hyperparameter selection}
We used five outer folds to obtain cross-fitted predictions of the outcome nuisance function $g_0(\bX)$ and the propensity score $\pi_0(\bX)$.
Within each outer training sample, hyperparameters of neural networks and random forests were selected using a single 80\%/20\% inner training-validation split.

For the neural network, the candidate grid included hidden-layer widths of 8, 16, or 32; dropout probabilities of 0.05 or 0.10; learning rates of $10^{-3}$ or $10^{-4}$; 500, 1,000, or 2,000 training epochs; and weight-decay values of $10^{-4}$ or $10^{-3}$.
Each candidate was fitted on the inner training sample and evaluated on the inner validation sample using the Poisson negative log-likelihood.
After selection, a new neural network was fitted using the complete outer training sample and used to obtain $\hat g_i$ for observations in the held-out outer fold.
The same cross-fitted $\hat g_i$ predictions were used for the DML estimators with logistic regression and random forest propensity scores.

The logistic regression propensity model required no hyperparameter tuning and was fitted on each complete outer training sample.
For the random forest, the candidate grid included 100 or 200 trees; maximum tree depths of 4 or 6; minimum split sizes of 2, 5, or 10; minimum leaf sizes of 2, 5, or 10.
At each split, the number of candidate predictors was determined by the square root of the total number of covariates.
Each candidate was evaluated using treatment log loss on the inner validation sample.
Similar to the outcome model, the selected random forest was refitted on the complete outer training sample and used to predict propensity scores in the held-out outer fold.
For numerical stability, the estimated propensity scores used by DML were truncated to $[10^{-6},1-10^{-6}]$.

\subsection{Inference for comparison methods}
For the AIPW estimators, we used a nonparametric bootstrap for their confidence intervals. 
The propensity scores for both AIPW estimators were obtained using five-fold cross-fitting.
We truncated the propensity scores from the logistic regression to $[0.01,0.99]$ and the RF propensity scores to $[10^{-6},1-10^{-6}]$.
We resampled the analytic data for $B=1,000$ bootstrap samples and refitted the outcome and propensity score models for each sample.
Percentile intervals were constructed for the AIPW estimates.
In addition, HC0 robust standard errors were used for the Poisson GLM model. 

\subsection{Distributions of inverse probability weights}
\label{sec:emp_weight_distribution}
Figure~\ref{fig:aipw_weight_distributions} presents the distributions of the raw and truncated inverse probability weights from logistic regression and random forests.

The first row shows that the prespecified truncation of propensity scores estimated using logistic regression reduced the magnitude of the most extreme inverse probability weights.
In contrast, the propensity scores estimated using random forests were not truncated, yet the resulting weights remained less extreme than those from logistic regression (second row).

\begin{figure}[H]
\centering
\includegraphics[width=0.85\textwidth]{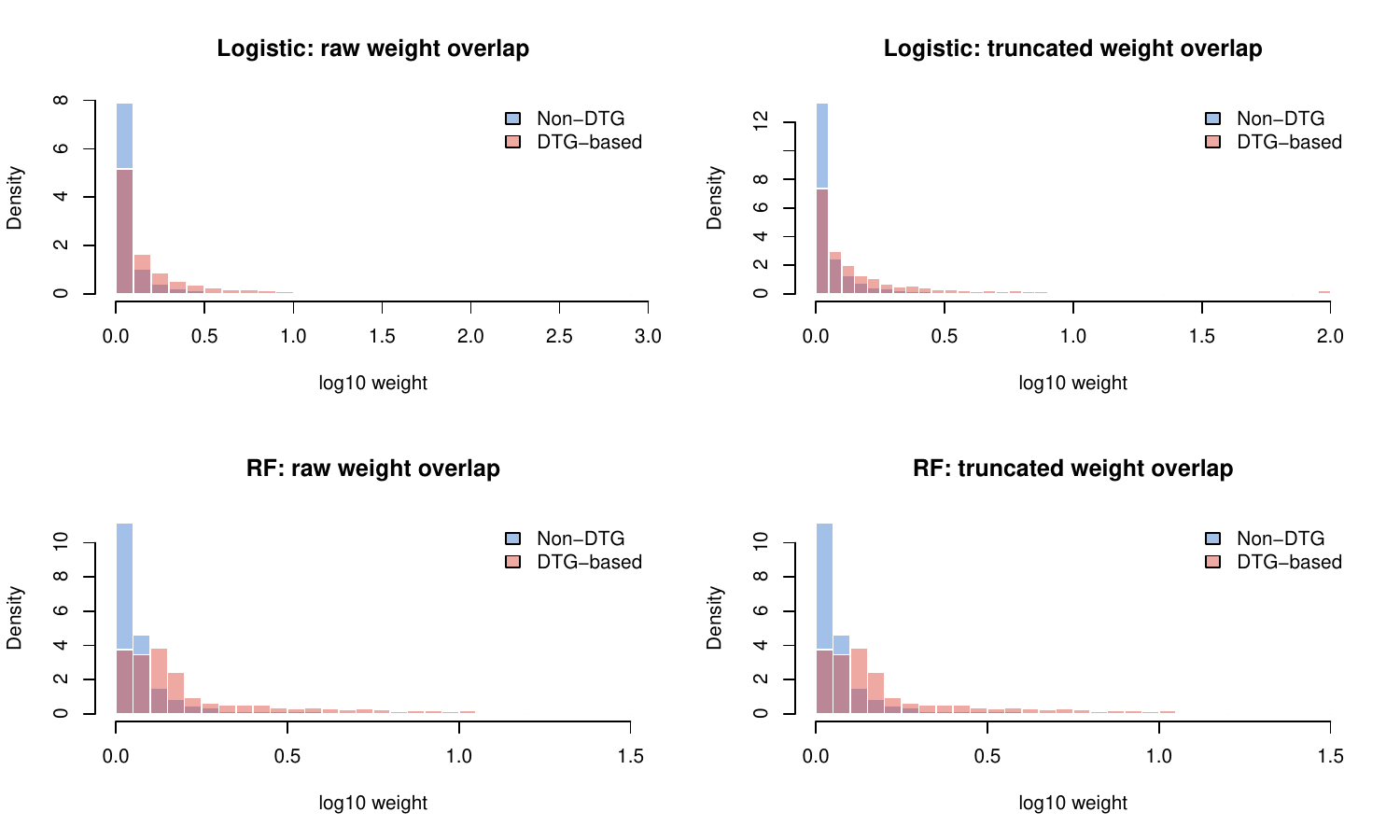}
\caption{\textbf{Inverse probability weight distributions.} The first row shows weights based on logistic regression, and the second row shows weights based on random forests. Within each row, the first panel shows the raw weights by ART regimen and their overlapping distributions, while the second panel shows the weights after propensity score truncation. Weights are displayed on the $\log_{10}$ scale; blue is non-DTG regimens and red is DTG-based regimens.}
\label{fig:aipw_weight_distributions}
\end{figure}

\end{document}